\documentclass[reprint,
superscriptaddress,showkeys,
nofootinbib,
 amsmath,amssymb,
 aps,
prx,
]{revtex4-2}

\usepackage{cmap}
\usepackage{graphicx}
\usepackage{dcolumn}
\usepackage{mathtools,physics,bbm,euscript,mathrsfs,enumitem,nicefrac,amsthm,amsfonts,graphicx,booktabs,float,amssymb,tcolorbox}
\usepackage{tikz}
\usetikzlibrary{arrows.meta, positioning, calc, intersections,decorations.markings,plotmarks}
\usepackage[T1]{fontenc}
\usepackage[utf8]{inputenc}

\usepackage[
colorlinks,bookmarks=true, citecolor=blue!90!black, urlcolor=blue!90!black,linkcolor=blue!90!black, pdfencoding=auto, psdextra]{hyperref}
\hypersetup{pdfauthor={Name}}

\usepackage[capitalize]{cleveref}
\newtheorem{theorem}{Theorem}[section]
\newtheorem{lemma}[theorem]{Lemma}

\newtheorem{proposition}[theorem]{Proposition}

\newtheorem{construction}{Construction}[section]
\theoremstyle{remark}
\newtheorem{definition}{Definition}[section]
\newtheorem{remark}{Remark}
\newtheorem{example}{Example}

\newcommand\nc\newcommand
\nc\bfa{{\boldsymbol a}}\nc\bfA{{\boldsymbol A}}\nc\cA{{\EuScript A}}
\nc\bfb{{\boldsymbol b}}\nc\bfB{{\boldsymbol B}}\nc\cB{{\EuScript B}}
\nc\bfc{{\boldsymbol c}}\nc\bfC{{\boldsymbol C}}\nc\cC{{\mathscr C}}\nc\uc{{\underline c}}
\nc\bfd{{\boldsymbol d}}\nc\bfD{{\boldsymbol D}}\nc\cD{{\mathscr D}}
\nc\bfe{{\boldsymbol e}}\nc\bfE{{\boldsymbol E}}\nc\cE{{\EuScript E}}\nc\ue{{\underline e}}
\nc\bff{{\boldsymbol f}}\nc\bfF{{\boldsymbol F}}\nc\cF{{\mathcal F}}\nc\uf{{\underline f}}
\nc\bfg{{\boldsymbol g}}\nc\bfG{{\boldsymbol G}}\nc\cG{{\EuScript G}}
\nc\bfh{{\boldsymbol h}}\nc\bfH{{\boldsymbol H}}\nc\cH{{\mathcal H}}\nc\uh{{\underline h}}
\nc\bfi{{\boldsymbol i}}\nc\bfI{{\boldsymbol I}}\nc\cI{{\mathcal I}}
\nc\bfj{{\boldsymbol j}}\nc\bfJ{{\boldsymbol J}}\nc\cJ{{\EuScript J}}
\nc\bfk{{\boldsymbol k}}\nc\bfK{{\boldsymbol K}}\nc\cK{{\EuScript K}}
\nc\bfl{{\boldsymbol l}}\nc\bfL{{\boldsymbol L}}\nc\cL{{\EuScript L}}
\nc\bfm{{\boldsymbol m}}\nc\bfM{{\boldsymbol M}}\nc\cM{{\EuScript M}}
\nc\bfn{{\boldsymfol n}}\nc\bfN{{\boldsymbol N}}\nc\cN{{\mathscr N}}\nc\un{{\underline n}}
\nc\bfo{{\boldsymbol o}}\nc\bfO{{\boldsymbol O}}\nc\cO{{\EuScript O}}
\nc\bfp{{\boldsymbol p}}\nc\bfP{{\boldsymbol P}}\nc\cP{{\EuScript P}}
\nc\bfq{{\boldsymbol q}}\nc\bfQ{{\boldsymbol Q}}\nc\cQ{{\EuScript Q}}
\nc\bfr{{\boldsymbol r}}\nc\bfR{{\boldsymbol R}}\nc\cR{{\EuScript R}}\nc\ur{{\underline r}}
\nc\bfs{{\boldsymbol s}}\nc\bfS{{\boldsymbol S}}\nc\cS{{\EuScript S}}
\nc\bft{{\boldsymbol t}}\nc\bfT{{\boldsymbol T}}\nc\cT{{\EuScript T}}
\nc\bfu{{\boldsymbol u}}\nc\bfU{{\boldsymbol U}}\nc\cU{{\EuScript U}}
\nc\bfv{{\boldsymbol v}}\nc\bfV{{\boldsymbol V}}\nc\cV{{\mathscr V}}
\nc\bfw{{\boldsymbol w}}\nc\bfW{{\boldsymbol W}}\nc\cW{{\mathscr W}}
\nc\bfx{{\boldsymbol x}}\nc\bfX{{\boldsymbol X}}\nc\cX{{\EuScript X}}\nc\ux{{\underline x}}
\nc\bfy{{\boldsymbol y}}\nc\bfY{{\boldsymbol Y}}\nc\cY{{\mathscr Y}}\nc\uy{{\underline y}}
\nc\bfz{{\boldsymbol z}}\nc\bfZ{{\boldsymbol Z}}\nc\cZ{{\EuScript Z}}
\nc{\1}{{\mathbbm{1}}}

\nc{\remove}[1]{}

\DeclareSymbolFont{bbold}{U}{bbold}{m}{n}
\DeclareSymbolFontAlphabet{\mathbbold}{bbold}

\DeclareMathOperator{\wt}{wt}
\DeclareMathOperator{\comp}{{\sf C}}

\newcommand{\R}{{\mathbb R}}

\newcommand{\Z}{{\mathbb Z}}

\nc\reals{{\mathbb R}}
\nc{\ff}{{\mathbb F}}
\nc{\PP}{{\mathbb P}}
\nc{\complex}{{\mathbb C}}

\newcommand{\ah}{{\hat{a}}}

\usepackage{xcolor}

\allowdisplaybreaks

\begin{document}

\title{Quantum Error Correction beyond \texorpdfstring{\(SU(2)\)}{}:\\
Spin, Bosonic, and Permutation-Invariant Codes from Convex Geometry}

\author{Arda Aydin}
 \email{aaydin@umd.edu}
\affiliation{ISR and Department of ECE, University of Maryland, College Park, MD 20742
}
\author{Victor V.\@ Albert}
\email{vva@umd.edu}
\affiliation{Joint Center for Quantum Information and Computer Science,
NIST/University of Maryland, College Park, MD 20742}
\author{Alexander Barg}%
 \email{abarg@umd.edu}
\affiliation{ISR and Department of ECE, University of Maryland, College Park, MD 20742
}
\affiliation{Joint Center for Quantum Information and Computer Science,
NIST/University of Maryland, College Park, MD 20742}

\begin{abstract}
We develop a framework for constructing quantum error-correcting codes and logical gates for three types of spaces --- composite permutation-invariant spaces of many qubits or qudits, composite constant-excitation Fock-state spaces of many bosonic modes, and monolithic nuclear state spaces of atoms, ions, and molecules.
By identifying all three spaces with discrete simplices and representations of the Lie group \(SU(q)\), we prove that many codes and their gates in \(SU(q)\) can be inter-converted between the three state spaces.
We construct \remove{new} code instances for all three spaces using classical \(\ell_1\) codes and Tverberg's theorem,
a classic result from convex geometry.
We obtain \remove{new} families of quantum codes with distance that scales almost linearly with the code length $N$ by constructing $\ell_1$ codes based on combinatorial patterns called Sidon sets and utilizing their Tverberg partitions. 
This compares favorably with the existing designs for all the
state spaces.
We present explicit constructions of codes with shorter length or lower total spin/excitation than known codes with similar parameters, \remove{new} bosonic codes with exotic Gaussian gates, as well as examples of short codes with distance larger than the known constructions.
\end{abstract}

           
\maketitle

 \tableofcontents

\section{\label{sec:Introduction} Introduction}

Quantum error correction, a necessary ingredient for realizing useful quantum algorithms, is the art of encoding quantum information into a strategically chosen subspace of an available state space so as to protect this information from noise.

Quantum state spaces come in many shapes and sizes.
For example, many-body, or \textit{composite}, qubit systems have associated with them a permutation-invariant (PI) subspace~\cite{harrow2013church,ruskaiExchange,ruskai-polatsek}, whose states are natural to realize using, e.g., collective atomic ensembles in cavity Quantum Electrodynamics (QED) (see, e.g., Refs.~\cite{chaudhury2007quantum,haas2014entangled,strobel2014fisher,lucke2014detecting,mcconnell2015entanglement,pezze2018quantum}).
On the other hand, nuclear state spaces of atoms, ions, or molecules house \textit{monolithic} spin-like spaces~\cite{gross} which exhibit long lifetimes yet can be fully controlled (see, e.g., recent works~\cite{asaad2020coherent,fernandez2024navigating,low2025control,debry2025error,yu2025schrodinger,ringbauer2022universal} and references therein).
{Artificial, or synthetic, spin spaces can even be engineered inside low-lying subspaces of a cavity \cite{roy2025synthetic,champion2025efficient}.

Both monolithic and composite spaces, the latter consisting of either qubits or bosonic modes, are currently under active investigation because a clear winner in the race to the first fault-tolerant quantum computer is yet to be determined, and because a range of different quantum platforms are likely to remain useful for other computing-adjacent applications.
As such, it is important to understand their ability to house quantum information in as simple and unified way as possible.
This task has proven difficult because conventional coding-theoretic constructions are not relevant outside of the many-qubit block code setting.

Progress to connect composite and monolithic state spaces has been made through Lie group theory. 
PI and spin state spaces are closely related in that they are both irreducible representations, or irreps, of the angular momentum group \(SU(2)\).
As such, their associated noise models and error-correcting codes --- PI codes and spin codes --- can be mapped into each other using properties of this group.
We extend these connections to the bosonic setting by identifying the state space of two bosonic modes of fixed excitation number with irreps of \(SU(2)\) using the Jordan-Schwinger map.

More importantly, we generalize code constructions for \textit{all three} state spaces by extending the group-theoretic connection to \(SU(q)\) for general \(q\). 
This yields relations among the spaces' noise models, codes, and logical gates.

For the case of spin codes, we point out a conceptually new way to interpret monolithic state spaces as irreps of \(SU(q)\), giving rise to new codes for \textit{existing} atomic, molecular, and ionic quantum platforms.
The geometry offered by this interpretation extends the power of monolithic systems since \(SU(q>2)\) codes can pack more logical information into the same-size space with minimal change in their protection.
Since nuclear manifolds are often completely controllable~\cite{asaad2020coherent,fernandez2024navigating,low2025control}, such codes are no less realizable than their \(SU(2)\) counterparts.
Moreover, mis-calibration or stray pulses associated with the \(SU(q)\) control unitaries should be compatible with the \(SU(q)\) geometry we introduce here.

For the bosonic case, we construct codes on constant-excitation Fock-state spaces of an arbitrary number of modes.
Complementing recent multimode code constructions~\cite{jain2024quantum,xu2024letting} that use coherent states, we obtain codes composed of finite sets of Fock states.
Mapping existing spin and PI codes with exotic transversal gates~\cite{gross,exoticGates}, we obtain new Fock state codes whose gates are realized using passive linear-optical transformations.

For the case of PI codes, we obtain new codes in the permutation-invariant subspace of multiple qudits by leveraging a classical family of \(\ell_1\) codes and established results in convex geometry.
Since all three state spaces are associated with irreps of \(SU(q)\), these PI codes can be converted to both spin and Fock state codes with similar error-correcting properties.

\section{Summary of technical results}
Permutation invariant (PI) codes form a class of quantum codes introduced in \cite{ruskaiExchange,ruskai-polatsek}. They are well-suited to recover encoded states from deletion errors \cite{ouyangEquivalence,hagiwaraDeletion,ouyangPI}.
A large family of qubit PI codes was recently constructed in \cite{aydin2023family}.

The authors of \cite{jain2024absorption} hinted at a relation between codes that correct photon absorption-emission errors (AE codes), a class of molecular codes that can be interpreted as spin codes, and PI codes. This connection was fully developed in \cite{aydin2025class}, where some of us showed that PI codes can be mapped to AE codes, giving rise to a family of AE codes hosted in a system with low angular momentum, which supports their efficient implementation and practical utility. The mapping of PI codes to AE codes, coupled with earlier results \cite{gross2,kubischta2023notsosecret} on a relation between spin codes of \cite{gross} and PI codes, also facilitated a link between the spin and AE code families  \cite{aydin2025class}. 

\begin{figure}[!t]
\includegraphics[width=0.95\columnwidth]{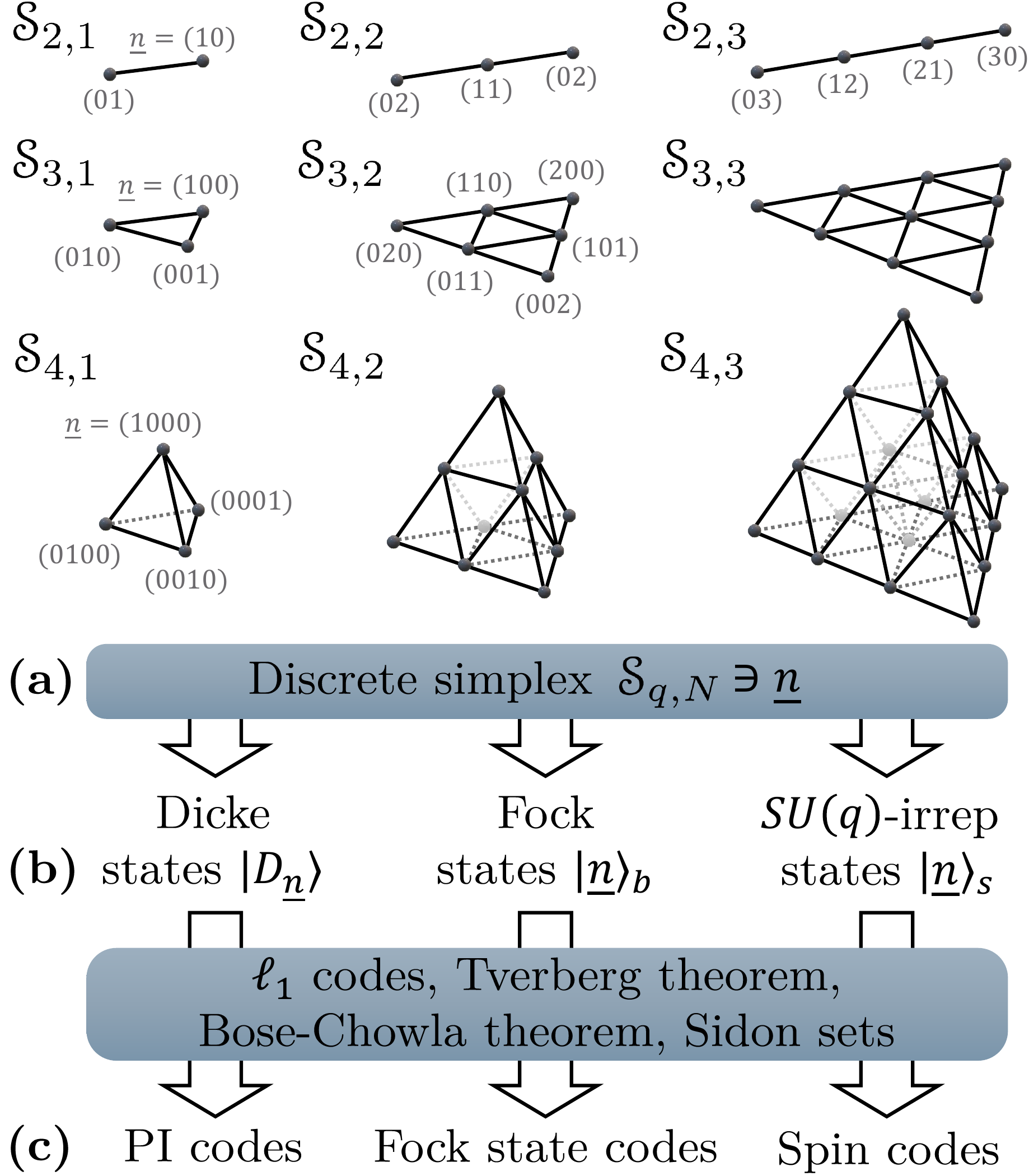}
\caption{
(a) Angular momentum ladders of spin-\(N/2\) irreps of \(SU(2)\) are in one-to-one correspondence with points of the discrete simplex \(\cS_{2,N}\).
This correspondence can be extended to one between completely symmetric \(SU(q)\)-irreps and higher-dimensional discrete simplices \(\cS_{q,N}\) (see \cref{subsec:js}).
(b) Simplices can be mapped into permutation-invariant (PI), bosonic, and spin spaces by associating simplex points with \(N\)-\(q\)udit Dicke states, \(N\)-excitation \(q\)-mode Fock states, and \(SU(q)\) ``spin'' states, respectively.
(c) This gives rise to new qudit PI, Fock state, and spin codes via classical $\ell_1$ codes and results from convex geometry (see \cref{sec: classical}).
}
\label{fig: SqN}
\end{figure}

Here we further connect qudit PI codes with Fock state and spin codes.
Our construction takes the following path. As one of the main results, we establish an equivalence between 
qudit PI codes, defined in \cite{ouyangQudit}, spin codes, and Fock state codes, supported by identifying all the three
systems with vertices of a discrete simplex, \cref{fig: SqN,fig: SqN1}. As a result,
we can seek constructions of codes within these families in a unified manner. 
\remove{Toward this end, we develop
an observation made earlier in \cite{aydin2023family} in the qubit case, namely, that the action of deletion errors on Dicke states can be expressed in a simple form. Since deletions are equivalent to erasures for PI codes, we can use deletion-correcting properties to estimate the code distance, so this approach forms a vehicle for constructing families of good PI codes. 
Here we further this line of thought, extending it to qudits as well as to powers of the deletion operators to match the photon loss errors in Fock state codes. Following the path of deletions rather
than Pauli errors again simplifies calculation of the distance of PI codes, this time in the qudit case. }

As our first step, we extend the Knill-Laflamme error-correcting conditions for qubit PI codes, whose specific form was found in \cite{aydin2023family}, to the case of general qudits, see \cref{fig: KL conditions}. We further show, in \cref{sec: Fock,sec:spin}, that similar-looking conditions support
error correction for multi-mode Fock state codes as well as for $SU(q)$ spin codes. 

At this point, the problem
of code construction reduces to finding solutions to a system of equations with respect to the basis coefficients of the codes.
To show that such solutions exist, we leverage the relation of the codes to the discrete simplex. Since its vertices
host classical $\ell_1$ codes, we derive sufficient conditions for the existence of quantum codes in terms of $\ell_1$ codes. 
To further link $\ell_1$ codes with the quantum error correction conditions, we
engage a classic result in convex geometry known as the Tverberg theorem, which shows that a sufficiently large point set in $\R^m$ can be partitioned into $K$ subsets such that their convex combinations have a nonempty intersection. Without explicitly mentioning the geometric link, this approach to the construction of PI codes was earlier hinted at in \cite{OuyangADCode,movassagh2024constructing}, and we fully develop it in this work.

The next step in our construction is finding families of classical $\ell_1$ codes with large distance. Starting
with a result in  \cite{kovacevic2018multisets} for such codes, we rely on combinatorial patterns known as Sidon ($B_t$) sets, whose
existence follows by another classical result, the Bose-Chowla theorem in additive number theory. This enables us to show
that there exist qudit PI codes, spin codes, and Fock state codes with distance that scales almost linearly with $N$ (the
length, total spin, and total excitation, respectively), \cref{sec: asymptotic}. In the final \cref{sec: Examples}, we further utilize the geometric connection to construct examples of spin, Fock state, and qudit PI codes as well as families of multi-mode covariant Fock state codes.

\begin{remark}[\sc Related results] After the completion of this paper we became aware that the general idea of using Tverberg partitions for the construction of quantum
codes has been considered in earlier literature following its introduction in \cite{knill2000theory}; see in particular \cite{bumgardner2012codes,cao2021higher}. While these works utilized it in an abstract setting of satisfying the Knill-Laflamme conditions for general quantum codes, we apply it in a concrete problem of constructing codes for several specific physical platforms. Starting with Tverberg partitions of classical $\ell_1$ codes with good parameters, we obtain improved asymptotics and specific examples of codes for the three interrelated systems considered here.

A related forthcoming work that develops new permutation invariant codes appears in \cite{vlad}.
\end{remark}

\begin{figure}[t]
\begin{center}\begin{tikzpicture}[node distance=4cm, every node/.style={font=\sffamily},
    circleNode/.style={circle, draw, minimum size=0.5in, align=center}, arrow/.style={-{Latex}, thick}, textnode/.style={font=\small}
  ]

 \node[circleNode] (L) at (-2,2) {$\ell_1$ codes};
  \node[circleNode] (A) at (0,0) {PI codes};
  \node[circleNode] (B) at (60:4) {Spin codes};
  \node[circleNode] (C) at (0:4) {Fock state\\ codes};


 \draw[dashed,arrow] (B) to node[pos=0.4,above right, textnode] {Prop.~\ref{prop: Spin to Fock}} (C);
 \draw[arrow,  bend right=10] (A) to node[below right, textnode] {Prop.~\ref{prop:equivalenceSpinPI}} (B);
 \draw[dashed,arrow,  bend right=10] (B) to node[above left, textnode] {\cref{lemma:ericandian}} (A);
\draw[arrow] (L) to node[pos=0.4,below left, textnode] {Thm.~\ref{prop:bound}} (A);

 \draw[arrow] (A) to node[pos=0.2,below right, textnode]{Prop.~\ref{prop:equivalence}} (C);
\end{tikzpicture}
\end{center}
\caption{Relations between $SU(q)$ code families. The solid arrows apply to all $q\ge 2$ and the dashed arrows apply only to $q=2$. The transformation Spin-to-PI in \cref{lemma:ericandian} is from \cite{IanEricDihedralIEEE}; the other transitions form new results.
}\label{fig: relations}
\end{figure}
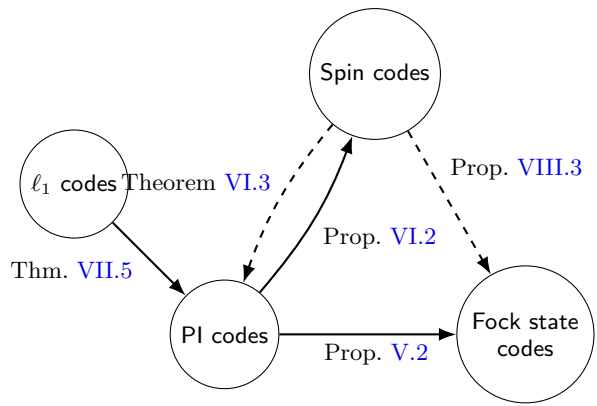

\section{Preliminaries}\label{sec:preliminaries} 
For a finite set $\cQ:=\{0,1,\dots,q-1\}$, let $S_q=\{ \ket{i}: i\in \cQ \}$ denote the standard orthonormal basis of the qudit space $\complex^q$, which 
describes a quantum system with $q$ levels. We begin with introducing elements of notation used
to write multi-dimensional qudit states using this basis. 
Let $U_q=\{ \bfu_i, i\in \cQ \}$ be the set of elementary basis vectors given by
  $$
  (\bfu_i)_j=\delta_{ij}, \quad j=0,1,\dots,q-1,
  $$
using the Kronecker delta notation. Then the vectors in $S_q$ can be written as $\ket{i}=\ket{\bfu_i}$ for all $i\in \cQ$.

Let $S(\complex^{q\otimes N})$ be the set of all density matrices of order $q^N$. For a vector $\bfx\in\cQ^N$, let 
$\bfx_{\sim i}:=(x_1,\ldots,x_{i-1},x_{i+1},\ldots,x_N),$ with the $i$-th entry missing. The {\em partial trace} of an 
 $q^N \times q^N$ matrix
  $$ 
  M=\sum_{\bfx,\bfy\in \cQ^N}m_{\bfx,\bfy}\ket{\bfx}\bra{\bfy},
  $$
is a mapping given by
\begin{align*}
    \Tr_i : S(\complex^{q\otimes N}) &\rightarrow S(\complex^{q\otimes (N-1)})\\
    M &\mapsto \sum_{\bfx,\bfy\in \cQ^N}m_{\bfx,\bfy}\Tr\left(\ket{x_i}\bra{y_i}\right)\ket{\bfx_{\sim i}}\bra{\bfy_{\sim i}}.
\end{align*}

We will use the notation $\un,\ue,\uf$ to refer to $q$-tuples of nonnegative integers, where for instance, $\un:=(n_0,n_1,\dots,n_{q-1})$.
For $N\in\Z_+,$ let
\begin{align}\label{eq:SimplexDef}
    \cS_{q,N}:=\Big\{\un\in \Z_0^q: \sum_{i=0}^{q-1} n_i=N\Big\},
\end{align} 
be a discrete simplex, i.e., the set of all integer partitions of $N$ into at most $q$ parts, see \cref{fig: SqN,{fig: SqN1},fig: S44}. 
Plainly,
$$|\cS_{q,N}|=\binom{N+q-1}{q-1}$$.

Multinomial coefficients are defined in a standard way:
\begin{align*}
    \binom{N}{\un}= \begin{cases}
        \frac{N!}{n_0!n_1!\ldots n_{q-1}!}, &\text{if }\un\in\cS_{q,N}, \\
        0 & \text{otherwise},
    \end{cases}
\end{align*}
with $0!=1$ by definition.

\begin{definition}\label{def:composition}
   For a $q$-ary string $\bfx \in \cQ^N$ let 
   $$
   n_i(\bfx)=|\{j\in [N]: x_j=i\}|, \quad i\in \cQ,
   $$
omitting the mention of $\bfx$ when it is clear from the context. Define the {\em composition}\footnote{In information theory, one often considers the probability distribution $\frac 1N \comp(\bfx)$, calling it the {\em type} of $\bfx$ \cite{csiszar1998method}.} of $\bfx$ as the tuple 
$\comp(\bfx)=(n_0,n_1,\ldots,n_{q-1})$ and note that $\sum_{i=0}^{q-1}n_i = N$ and 
  $$|\{\bfx: \comp(\bfx)=\un\}|=\binom N\un.
  $$
  For $q=2$, the composition $\comp(\bfx)$ is fully determined by
the Hamming weight $n_1(\bfx)$, denoted below as $\wt(\bfx)$.
\end{definition}  

\begin{definition} A 
\textit{Dicke state} is a linear combination of all qudit states with the same composition. We define
 \begin{equation}\label{eq: Dicke}
    \ket{D_\un}= \frac{1}{\sqrt{\binom{N}{\un}}}\sum_{\substack{\bfx\in\cQ^N \\ \comp(\bfx)=\un}} \ket{\bfx}.
\end{equation}
For $q=2$ this reduces to the more standard notion of a Dicke state, $\ket {D_w}\propto \sum_{\begin{substack}{\bfx\in\{0,1\}^N\\ \wt(\bfx)=w}\end{substack}}\ket\bfx$ \cite{Dicke1,sekhar2020actual}.
\end{definition}

\begin{figure}
\includegraphics[width=.9\linewidth]{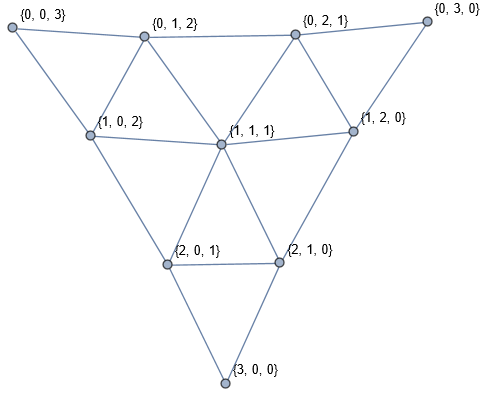}
\caption{Graph of the discrete simplex $\cS_{3,3}$, where the vertices $u,v$ are connected iff $d_1(u,v)=1$; see \cref{eq: d1}. This space is used to define qudit PI codes and Fock state codes (as the index set of basis coefficients), as well as classical $\ell_1$ codes of \cref{sec: classical} (as a metric space).}
\label{fig: SqN1}
\end{figure}

To give an example, if $q=3$ and $\bfx=(0,2,2)$, then $\comp(\bfx)=(1,0,2)$ and 
\begin{align*}
    \ket{D_{(1,0,2)}}=\frac{1}{\sqrt{3}}\left( \ket{022}+\ket{202}+\ket{220} \right).
\end{align*}

{\em Error correction conditions:} 
The general {Knill-Laflamme} conditions give a criterion for error correction for quantum codes
used on a quantum communication channel represented by a set of Kraus operators $\{A_a\}$. Somewhat informally, a code corrects errors in this set if and only if $\bra {\bfc_i}A_a^\dag A_b\ket{\bfc_j}=0$ and $\bra {\bfc_i}A_a^\dag A_b\ket{\bfc_i}=g_{ab}$ for all $i\ne j$ and all $a,b$, where $(\bfc_i)_i$ is 
an orthonormal basis of the code and $g_{ab}\in \complex$ is some set of constants. These conditions are called orthogonality and non-deformation, respectively. 

The Knill-Laflamme conditions for qubit PI codes were rendered in an explicit form in our recent work \cite{aydin2023family} based on the approach via deletion errors. Here we extend them to general qudit PI codes and Fock state codes to quantify the error
correcting performance of the codes. 

The \(i\)th code basis state of a PI code of logical dimension \(K\) is a superposition of Dicke states \(|D_\un\rangle\) (see Def.~\ref{Def:PICode}).
Denoting the state's basis expansion coefficients by $\alpha_{\un}^{(i)}\in \mathbb{C}$, the Knill-Laflamme conditions translate to equations (\hyperlink{C3}{C3}), (\hyperlink{C4}{C4}) shown in Fig.~\ref{fig: KL conditions}. This fact will be shown in the context of the three main code classes considered in the paper, with proofs provided in Thm.~\ref{Theo:PIMainTheorem} (PI codes), Thm.~\ref{Theo:BosonMainTheoremAD} (Fock state codes), and
Thm.~\ref{Theo:SpinMainTheorem} (spin codes).
\begin{figure}
\begin{tcolorbox}[width=1.05\linewidth, left=1mm, right=1mm, boxsep=0pt, boxrule=0.5pt, sharp corners=all, colback=white!95!black]  
\begin{align*}
        &\hypertarget{C1}{\text{\rm (C1)}}\sum_{\un\in\cS_{q,N}}(\alpha^{(i)}_{\un})^*\alpha^{(j)}_{\un}=0\\
        &\hypertarget{C2}{\text{\rm (C2)}}\sum_{\un\in\cS_{q,N}}|\alpha^{(i)}_{\un}|^2=\sum_{\un\in\cS_{q,N}}|\alpha^{(j)}_{\un}|^2\\
        &\hypertarget{C3}{\text{\rm (C3)}}\sum_{\un\in\cS_{q,N}}(\alpha^{(i)}_{\un})^*\alpha^{(j)}_{\un-\ue+\uf}\frac{\binom{N-t}{\un-\ue}}{\sqrt{\binom{N}{\un}\binom{N}{\un-\ue+\uf}}}=0\\
        &\hypertarget{C4}{\text{\rm (C4)}}\sum_{\un\in\cS_{q,N}}\!\!\!\!\left((\alpha^{(i)}_{\un})^*\alpha^{(i)}_{\un-\ue+\uf}-(\alpha^{(j)}_{\un})^*\alpha^{(j)}_{\un-\ue+\uf}\right)\!\frac{\binom{N-t}{\un-\ue}}{\sqrt{\binom{N}{\un}\binom{N}{\un-\ue+\uf}}}=0
    \end{align*}
    Here $\ue,\uf\in\cS_{q,t}$ where $t$ is the number of errors that a qudit PI code can detect, and 
$i,j\in \{0,1,\dots,K-1\}, i\ne j$.
\end{tcolorbox}
\caption{Conditions for construction and error correction with qudit PI codes. Eqns.~(C3), (C4)
represent the KL conditions. We relate similar conditions for spin and Fock state codes to these using multinomial coefficient identities.
}\label{fig: KL conditions}
\end{figure}

\section{PI codes}

Our construction starts with PI codes \cite{ruskaiExchange,ruskai-polatsek,ouyangPI,ouyangQudit}, which were earlier shown to yield codes for the amplitude damping and related error models in \cite{aydin2025class,kubischta2023notsosecret}. We proceed with the definition.
\begin{definition}\label{Def:PICode}
    A $K$-dimensional qudit PI code $\cQ_{PI}$ is defined by the set of basis vectors 
    \begin{align} \label{eq: basis vectors}
        \ket{\bfc_i} = \sum_{\un \in \cS_{q,N}}\alpha^{(i)}_{\un}\ket{D_{\un}}, \quad i=0,1,\ldots,K-1.
    \end{align}
 Here $ \ket{D_{\un}} $ is a Dicke state given in \cref{eq: Dicke} and $ \alpha_{\un}^{(i)} \in \complex $ $i=0,1,\ldots,K-1$ is a coefficient set. We also call the number $N$ the {\em length} of the code, and if the coefficients $\alpha_\un^{(i)}$ are real,
 we call $\cQ_{PI}$ a {\em real code}.
\end{definition}

To define the distance of a PI code, we use the approach via deletions. Recall that an {\em erasure} is essentially a deletion with a known location. By symmetry, for PI codes, deletions are equivalent to erasures as deleting $t$ qudits whose locations are unknown is the same as deleting the first $t$ qudits: the resulting state in both cases is the same. Hence, if a PI code corrects the set of errors in \cref{eq: KrausSetForDelChannel}, it can also correct $t$ erasures. A code that corrects $t$ erasures clearly has distance at least $t+1$. With this in mind, we say that the code has {\em distance} at least $t+1$ if it corrects all errors in the set given by \cref{eq: KrausSetForDelChannel}.
Given an $q$-qudit PI code of length $N$, dimension $K$, and distance $d$, we call it and $(N,K,q,d)$ PI code.

\subsection{Quantum deletion channel} In quantum coding theory, the deletion channel is the partial trace operation on a set of positions that are unknown. The formal definition is given as follows:
\begin{definition}
    Let $t\le N$ be a positive integer. Consider a quantum state $\rho \in S\left(\complex^{q\otimes n}\right)$ and a set of indices $I = \{i_1, i_2, \ldots, i_t\} \subset \{1, 2, \ldots, n\}$, where $i_1< i_2<\ldots<i_t$. The \textit{$t$-deletion error} corresponding to a subset $I$ is a composition of partial trace operations on its elements:
    \begin{align*}
        D_I^N(\rho) := \Tr_{i_1}\circ \ldots \circ \Tr_{i_t}(\rho).
    \end{align*}
\end{definition}
A {\em $t$-deletion channel} operator is a convex combination of all $t$-deletion errors, i.e.,
\begin{align*}
    \operatorname{Del}_t^N(\rho) = \sum_{I:|I|=t}p_ID_I^N(\rho),
\end{align*}
where $p_I$ is a probability distribution supported on the deleted set $I$.

We will need an explicit form of the action of this channel on a Dicke state. 
For this, we use the Kraus decomposition of the $t$-deletion channel, which appears in \cite{ouyangEquivalence}. 
Let $x\in \cQ^N$ and $\ket j\in S_q$. 
A deletion error of type $j$ acting on the $i$-th qudit of a pure state $\ket \bfx, 
\bfx \in \cQ^N$ is written as
\begin{align*}
    \bfK^{(N)}_{j,i} \ket{\bfx} = \bra{j}\ket{x_i}\ket{\bfx_{\sim i}}, \quad i=1,\dots,N.
\end{align*}
 Let $I$ be a $t$-subset as above and let $J=\{j_1, j_2, \ldots, j_t\}$, where $j_k\in\cQ$ for $k=1,2,\ldots,t$. 
Then define the joint deletion error
\begin{align*}
    \bfK^{(N)}_{J,I}: = \bfK^{(N-t+1)}_{j_1,i_1}\circ\bfK^{(N-t+2)}_{j_2,i_2}\circ\dots\circ\bfK^{(N)}_{j_t,i_t},
\end{align*}
where the qudit $i_k$ is subjected to a deletion of type $j_k$ for all $k$.
Note that this operator can also be expressed as the tensor product 
  $$
  \bfK^{(N)}_{J,I}=\bfK_1\otimes \bfK_2\otimes \ldots \otimes\bfK_N,
  $$
where for all $l=1,\dots,N$, $K_l=\text{Id}$ for $l\not\in I$ and $K_l=\bra{j_k}$ if $l=i_k$, $k=1,\dots,t$.
See \cite[Lemma 3.1]{ouyangEquivalence}  for details.

Using the operators $K_{I,J}$, we can write a Kraus decomposition of the deletion channel \cite{ouyangEquivalence}:
\begin{align*}
     \operatorname{Del}_t^N(\rho) = \sum_{I,J}p_I \bfK^{(N)}_{J,I}\rho \bfK^{(N)\dagger}_{J,I},
\end{align*}
where $p_I$ is a probability distribution.

\subsection{Quantum deletions acting on Dicke states} \label{sec: deletions}
Action of deletions on qubit Dicke states was previously analyzed in \cite{aydin2023family}. In this section, we extend those results to  
the case of qudits.

\begin{lemma}\label{Lemma:DelOperatorAction}
    Let $ q\geq 2$ and $N\ge 0$ be integers. For a Dicke state $ \ket{D_\un},\un \in \cS_{q,N} $, all $i\in\{1,2,\ldots,N\}$ and $ j\in \{0,1,\ldots,q-1\} $,
    \begin{align}
        \bfK_{j,i}^{(N)}\ket{D_{\un}}= \sqrt{\frac{\binom{N-1}{\un-\bfu_j}}{\binom{N}{\un}}}\ket{D_{\un-\bfu_j}}, \label{eq: ji}
    \end{align}
    where $\bfu_j\in U_q$.
\end{lemma}
\begin{proof}
Let $\bfx \in \cQ^N$. Consider the unnormalized Dicke state $ \ket{H_\un}=\sum_{\bfx: \comp(\bfx)=\un} \ket{\bfx}$. Fix an $i\in\{1,2,\dots,N\}$.
As a result of applying $ \bfK_{j,i}^{(N)}$ to this state, the terms (states) $\ket\bfx$ with $x_i\neq j$ vanish, and the terms $\ket \bfx$ with $x_i=j$ are subjected to deletion, i.e., mapped to $\bfx_{\sim i}$.
Therefore, the number of $j$'s in the remaining states decreases by one, and as a result, the composition
 of the terms becomes $\un^\prime=(n_0,\dots,n_{j-1},n_j-1,n_{j+1},\ldots,n_{q-1})=\un-\bfu_j$. Succinctly, we have
    \begin{align*}
       \bfK_{j,i}^{(N)}\ket{H_\un} = \ket{H_{\un-\bfu_j}}.
    \end{align*}
Normalizing, we obtain \cref{eq: ji}. 
\end{proof}
Because of the symmetry of Dicke states, intuition suggests that the order of applying the deletion operators does not matter. 
The next lemma establishes this property rigorously.
\begin{lemma}\label{Lemma:Commutation}
     Let $\un\in\cS_{q,N}$ and $\ket{D_\un}$ be a Dicke state. Then for any $ j_1,j_2\in \{ 0,1,\ldots,q-1 \} $, $i_1,i_1^\prime \in \{ 1,2,\ldots,N \}$, and $i_2,i_2^\prime \in \{1,2,\ldots,N-1\}$, we have
    \begin{align*}
        \bfK^{(N-1)}_{j_1,i_2}\bfK^{(N)}_{j_2,i_1}\ket{D_\un} = \bfK^{(N-1)}_{j_2,i_2^\prime}\bfK^{(N)}_{j_1,i_1^\prime}\ket{D_\un} .
    \end{align*}
\end{lemma}
\begin{proof}
    Using \cref{Lemma:DelOperatorAction},
    \begin{align*}
        \bfK^{(N-1)}_{j_1,i_2}\bfK^{(N)}_{j_2,i_1}\ket{D_\un} &= \sqrt{\frac{\binom{N-1}{\un-\bfu_{j_2}}}{\binom{N}{\un}}}\bfK^{(N-1)}_{j_1,i_2}\ket{D_{\un-\bfu_{j_2}}}\\
        &=\sqrt{\frac{\binom{N-2}{\un-\bfu_{j_1}-\bfu_{j_2}}}{\binom{N}{\un}}}\ket{D_{\un-\bfu_{j_1}-\bfu_{j_2}}}.
    \end{align*}
    Similarly,
    \begin{align*}
        \bfK^{(N-1)}_{j_2,i_2^\prime}\bfK^{(N)}_{j_1,i_1^\prime}\ket{D_\un} &= \sqrt{\frac{\binom{N-1}{\un-\bfu_{j_1}}}{\binom{N}{\un}}}\bfK^{(N-1)}_{j_2,i_2}\ket{D_{\un-\bfu_{j_1}}}\\
        &=\sqrt{\frac{\binom{N-2}{\un-\bfu_{j_1}-\bfu_{j_2}}}{\binom{N}{\un}}}\ket{D_{\un-\bfu_{j_1}-\bfu_{j_2}}}\qedhere
    \end{align*}
\end{proof}
Again by symmetry of Dicke states, our arguments below will not depend on the location of the deleted qudits once their number is fixed.
This allows us to simplify the notation, writing $\bfK_j$ for a deletion of type $j$, without specifying the location $i$. Further, when the dimension $N$ is clear from the context, we also omit it.

A ``power'' of the deletion operator is defined through repeated application to states of decreasing
dimension. Let $e\ge 1, j\in\cQ,$ and $i_k\in\{ 1,2,\ldots,N-k+1 \}$ for $k=1,2,\dots,e.$  Define
    $$
    ( \bfK_j )^e\ket{D_\un} := \bfK^{N-e+1}_{j,i_e}\ldots\bfK^{N-1}_{j,i_2}\bfK^{N}_{j,i_1}\ket{D_{\un}}.
    $$
An explicit expression for $( \bfK_j )^e$ is derived next.
\begin{lemma}\label{Lemma:PowerOperator}
  Consider a Dicke state $\ket{D_\un},  \un\in\cS_{q,N}$. For any $j\in\cQ$, we have
    \begin{align*}
        ( \bfK_j )^e\ket{D_\un} 
        &=\sqrt{\frac{\binom{N-e}{\un-\ue_j}}{\binom{N}{\un}}}\ket{D_{\un-\ue_j}}.
    \end{align*}
    Here $\ue_j = e\bfu_j$, where $\bfu_j\in U_q$. 
\end{lemma}
\begin{proof}
    Using \cref{Lemma:DelOperatorAction},
    \begin{align*}
         \left( \bfK_j \right)^e\ket{D_\un} &= \sqrt{\frac{\binom{N-e}{\un-\ue_j}\ldots\binom{N-2}{\un-\bfu_j-\bfu_j}\binom{N-1}{\un-\bfu_j}}{\binom{N-e+1}{\un-\ue_j+\bfu_j}\ldots\binom{N-1}{\un-\bfu_j}\binom{N}{\un}}}\ket{D_{\un-\ue_j}}\\
         &=\sqrt{\frac{\binom{N-e}{\un-\ue_j}}{\binom{N}{\un}}}\ket{D_{\un-\ue_j}}. \qedhere
    \end{align*}
\end{proof}

\textit{Kraus set for the $t$-deletion channel:} We are now in a position to describe the Kraus set of a $t$-deletion channel for PI states. Let $e_0, e_1, \ldots, e_{q-1}$ be nonnegative integers. A deletion error labeled by the tuple
\begin{align}\label{eq:DefErrorVector}
    \ue := e_0\bfu_0+e_1\bfu_1+\ldots+e_{q-1}\bfu_{q-1},
\end{align}
meaning that a deletion of type $j$ ``occurred $e_j$ times'', i.e., by application of $(\bfK_j)^{e_j}$.
Note that the collection of error tuples of type \cref{eq:DefErrorVector} is precisely the discrete simplex $\cS_{q,t}$, cf. \cref{eq:SimplexDef}. The deletion operator associated with the error $\ue\in\cS_{q,t}$ is
\begin{align*}
    \mathbb{E}_\ue := (\bfK_0)^{e_0}(\bfK_1)^{e_1}\ldots (\bfK_{q-1})^{e_{q-1}}
\end{align*}
Using \cref{Lemma:DelOperatorAction,Lemma:Commutation,Lemma:PowerOperator}, we can write the Kraus set of the $t$-deletion channel for PI states as follows: 
\begin{align}\label{eq: KrausSetForDelChannel}
    \cE_{q,t} = \left\{\mathbb{E}_\ue :\ue \in \cS_{q,t} \right\}
\end{align}

The action of $\mathbb{E}_\ue$ on Dicke states takes a surprisingly simple form.
\begin{lemma}\label{Lemma:KrausOperatorAction}
    Let $\mathbb{E}_\ue\in \cE_{q,t}$ be a Kraus operator acting on a Dicke state $\ket{D_\un}$. We have
    \begin{align*}
        \mathbb{E}_\ue\ket{D_\un} = \sqrt{\frac{\binom{N-t}{\un-\ue}}{\binom{N}{\un}}}\ket{D_{\un-\ue}}.
    \end{align*}
\end{lemma}
\begin{proof}
    Using \cref{Lemma:PowerOperator}, 
    \begin{align*}
        \mathbb{E}_\ue\ket{D_\un}&=(\bfK_0)^{e_0}(\bfK_1)^{e_1}\ldots (\bfK_{q-1})^{e_{q-1}}\ket{D_\un}\\
        &=\sqrt{\frac{\binom{N-e_0-\ldots-e_{q-1}}{\un-\ue_0-\ldots-\ue_{q-1}}}{\binom{N-e_1-\ldots-e_{q-1}}{\un-\ue_1-\ldots-\ue_{q-1}}}}\\
        &\hspace*{.3in}\ldots\sqrt{\frac{\binom{N-e_{q-2}-e_{q-1}}{\un-\ue_{q-2}-\ue_{q-1}}}{\binom{N-e_{q-1}}{\un-\ue_{q-1}}}}\sqrt{\frac{\binom{N-e_{q-1}}{\un-\ue_{q-1}}}{\binom{N}{\un}}}\ket{D_{\un-\ue}}. \qedhere
    \end{align*}
\end{proof}
We will use this lemma to establish the error correction conditions for qudit PI codes in the next section.

\subsection{Error correction conditions for PI codes} 
Using the approach via the deletion errors in \cref{sec: deletions}, we show the error correction conditions for qudit PI codes.
\begin{theorem}\label{Theo:PIMainTheorem}
    Let $\cQ_{PI}$ be an $(N,K,q,d)$ PI code. Suppose that the coefficients in \cref{eq: basis vectors} satisfy $\alpha_{\un}^{(i)}\in \mathbb{C}$ for all $\un\in \cS_{q,N}, i=0,1,\ldots K-1$. Then $\cQ_{PI}$ has distance $d=t+1$ if and only if conditions {\rm (\hyperlink{C1}{C1})--(\hyperlink{C4}{C4})}
    hold for all distinct $i,j\in\{0,1,\ldots,K-1\}$ and all $\ue,\uf \in \cS_{q,t}$.
\end{theorem}
\begin{proof}
   Conditions (\hyperlink{C1}{C1}) and (\hyperlink{C2}{C2}) are necessary to ensure that the basis for the code space is orthonormal. We need to show that conditions (\hyperlink{C3}{C3}) and (\hyperlink{C4}{C4}) are equivalent to the Knill-Laflamme conditions for the qudit PI code with distance $d=t+1$.

Let $\mathbb{E}_{\ue},\mathbb{E}_{\uf}\in\cE_{q,t}$. By \cref{Lemma:KrausOperatorAction}, 
    \begin{align*}
        \bra{\bfc_i}\mathbb{E}_\ue^\dagger\mathbb{E}_{\uf}\ket{\bfc_j} &= \sum_{\un}\sum_{\un^\prime}(\alpha^{(i)}_{\un})^*\alpha^{(j)}_{\un^\prime}\sqrt{\frac{\binom{N-t}{\un-\ue}\binom{N-t}{\un^\prime-\uf}}{\binom{N}{\un}\binom{N}{\un^\prime}}}\delta_{\un-\ue,\un^\prime-\uf}\\
        &=\sum_{\un}(\alpha^{(i)}_{\un})^*\alpha^{(j)}_{\un-\ue+\uf}\frac{\binom{N-t}{\un-\ue}}{\sqrt{\binom{N}{\un}\binom{N}{\un-\ue+\uf}}},
    \end{align*}
    where we used orthonormality $\bra{D_{\un-\ue}}\ket{D_{\un^\prime-\uf}}=\delta_{\un-\ue,\un^\prime-\uf}$ (here $\delta_{\cdot,\cdot}$ is the Kronecker delta). Hence, Condition (\hyperlink{C3}{C3}) is equivalent to the orthogonality part of the Knill-Laflamme conditions. By repeating the same sequence of steps, one can show that Condition (\hyperlink{C4}{C4}) is equivalent to the non-deformation 
    part of the Knill-Laflamme conditions.
\end{proof}

\section{Fock state codes}\label{sec: Fock}

Spin codes and PI codes both lie within irreducible representations of \(SU(q)\).
We show that similar codes can also be constructed in a third setting, namely, the space of two or more bosonic (a.k.a.~continuous-variable) modes~\cite{albert2025bosonic,fabre2020modes,terhal2020towards}.
A bosonic state space is spanned by Fock states, and the subspace of two-mode Fock states with the same total excitation number houses irreducible representations of \(SU(2)\) (represented by the set of passive Gaussian transformations).
We use this and other connections to construct new bosonic codes supported on multi-mode Fock-state subspaces with the same total excitation.

A $K$-dimensional $q$-mode bosonic {\em Fock state code} is defined by the following basis:
\begin{definition}\label{Def:FockStateCodes}
    Let $\un \in \cS_{q,N}$ and let
    $$
    \ket{\un}_b:=\ket{n_0}_b\otimes \ket{n_1}_b\otimes\ldots\otimes\ket{n_{q-1}}_b,
    $$
where each $\ket{n_i}_b$ is a single-mode Fock state. A $K$-dimensional $q$-mode Fock state code $\cQ_{f}$ with total excitation $N$ is a $\complex$-linear space with the basis
    \begin{align}\label{eq:DefFockStateCodes}
        \ket{\bfc_i} = \sum_{\un\in\cS_{q,N}}\alpha^{(i)}_{\un}\ket{\un}_b,\quad i=0,1,\ldots,K-1.
    \end{align}
As before, we call $\cQ_f$ a {\em real code} if the coefficients $\alpha^{(i)}_{\un}$ are real. 
\end{definition}
Codes in this definition are often called {\em constant-excitation} codes because they are formed of states $\ket\un_b$ with a fixed $\ell_1$ norm of $\un$. Only such codes will be considered below.

Our notation $\ket{\cdot}_b$ indicates that the argument is a single- or multi-mode Fock state, which helps to distinguish it from qudit states, and we also extend this usage to inner products.

\subsection{Amplitude damping errors}

The prevailing error model in bosonic platforms is described by amplitude damping, which models leakage of energy carriers (photons or phonons) out of the system \cite{michael2016new,albert2018performance}.
The strength of the noise is governed by a real positive ``loss rate'' parameter \(\gamma\).
As the loss rate increases, the average excitation number of the input state decreases, with all states decaying to the all-zero, or vacuum, Fock state \(\ket{\underline{0}}_b\) as \(\gamma\to 1\).

Suppose that we have a code $\cQ$ defined by the basis states $ \{\ket{\bfc_i}\}_i$. For a quantum channel with the Kraus set $\cE$, the {\em code fidelity} is defined as
\begin{align*}
    \cF_{\cE,\cQ} = \min_{\psi_{\text{in}}} \sum_{\hat{E}\in\cE}\bra{\psi_{\text{in}}}\hat{E}^\dagger\hat{E}\ket{\psi_{\text{in}}}, 
\end{align*}
where the input state is
\begin{align*}
    \ket{\psi_{\text{in}}} = \sum_{i=0}^{K-1}a_i\ket{\bfc_i}. 
\end{align*}
Code fidelity is a measure of the distance between the input state and the recovered state, so the higher the fidelity value, the more error protection does the code offer.

Kraus operators for amplitude damping act on a single-mode Fock state according to the 
following pattern:
\begin{align}\label{eq:ADErrorAction}
    \hat{A}_x\ket{j}=\sqrt{(1-\gamma)^{j-x}\gamma^x}\sqrt{\binom{j}{x}}\ket{j-x},\quad j\geq x.
\end{align}

For $\ue\in\cS_{q,r}$, define the error operator $\hat{A}_{\ue}$ in the following way:
\begin{align*}
    \hat{A}_{\ue} = \hat{A}_{e_0}\otimes\hat{A}_{e_1}\otimes\ldots\otimes\hat{A}_{e_{q-1}}.
\end{align*}
Varying $\ue$, we obtain an error set
\begin{align}\label{eq:DefADErrorSet}
    \cE_t^{AD} = \{ \hat{A}_{\ue}: \ue\in\bar{\cS}_{q,t} \}, 
\end{align}
where
$\bar{\cS}_{q,t} = \bigcup_{1\leq r\leq t}\cS_{q,r}.$
\begin{definition}
    We say that a $q$-mode Fock state code $\cQ_f$ corrects $t$-AD errors if \remove{, for some constant $A>0$,} it satisfies the KL conditions for the error set in Eq. \eqref{eq:DefADErrorSet}.
\remove{\begin{align*}
    \cF_{\cE_t^{AD},\cQ_f}=1-A\gamma^{t+1}+\mathcal{O}(\gamma^{t+2}),
\end{align*}}
\end{definition}

In \cite{chuangADCode}, it is shown that \remove{if} for a constant-excitation Fock state code, \remove{\cref{eq:DefFockStateCodes}, satisfies the Knill--Laflamme conditions for the error set $\cE_t^{AD}$, then}
\begin{align}\label{eq: fidelity}
    \cF_{\cE_t^{AD},\cQ_f}= 1 - \binom{N}{t+1}\gamma^{t+1} + \mathcal{O}(\gamma^{t+2}).
\end{align}

In \cite{OuyangADCode}, the authors show that the entanglement fidelity also behaves in the same way. Note that both the worst-case and entanglement fidelity depend only on the total excitation $N$ and the multiplicity of errors $t$; in particular, they are independent of the number of modes $q$. 
This further affirms our goal to design codes with low total excitation.

\subsection{Error correction conditions}

In the following theorem, we show that the error-correction conditions for $q$-mode Fock state codes can be related to those of PI codes via the Vandermonde convolution.

\begin{theorem}\label{Theo:BosonMainTheoremAD}
Let $\cQ_f$ be a $K$-dimensional Fock state code with total excitation $N$ defined by the basis coefficients $\alpha_\un^{(i)}, \un\in \cS_{q,N}, 0\le i\le K-1$. Then $\cQ_f$ corrects $t$-AD errors if it satisfies conditions {\rm (\hyperlink{C1}{C1})--(\hyperlink{C4}{C4})}.
\end{theorem}
\begin{proof}
    Let $\un\in \cS_{q,N}$ and $\ur\in\cS_{q,r}$, where $r\leq t$. Then, using \cref{eq:ADErrorAction}
    and \cref{Lemma: multinomial} in succession, we obtain
    \begin{align*}
        \hat{A}_{\ur}\ket{\un}_b& = \sqrt{(1-\gamma)^{N-r}\gamma^r}\sqrt{\binom{n_0}{r_0}\ldots\binom{n_{q-1}}{r_{q-1}}}\ket{\un-\ur}_b\\
    &=  \sqrt{(1-\gamma)^{N-r}\gamma^r}\sqrt{\frac{\binom{N}{r}\binom{r}{\ur}\binom{N-r}{\un-\ur}}{\binom{N}{\un}}}\ket{\un-\ur}_b.   
\end{align*}
Let $r,r^\prime \leq t$ be nonnegative integers. Let $\ur\in\cS_{q,r}$ and $\ur^\prime\in\cS_{q,r^\prime}$, then the inner product
\begin{multline*}
    \bra{\bfc_i}\hat{A}_\ur^\dagger\hat{A}_{\ur^\prime}\ket{\bfc_j}= \\ 
    \sqrt{(1-\gamma)^{2N-r-r^\prime}\gamma^{r+r^\prime}\binom{N}{r}\binom{N}{r^\prime}\binom{r}{\ur}\binom{r^\prime}{\ur^\prime}} \\
    \times \sum_{\un}\sum_{\un^\prime}(\alpha^{(i)}_{\un})^*\alpha^{(j)}_{\un^\prime}\sqrt{\frac{\binom{N-r}{\un-\ur}\binom{N-r^\prime}{\un^\prime-\ur^\prime}}{\binom{N}{\un}\binom{N}{\un^\prime}}}\bra{\un-\ur}\ket{\un^\prime-\ur^\prime}_b
\end{multline*}
gives a nonzero value only when $\un-\ur=\un^\prime-\ur^\prime$, but this can hold only if the $\ell_1$ norms of $\ur,\ur'$ are equal, i.e., $r=r^\prime$. Using this, we further find
\begin{multline}\label{eq:InnerProd}
\bra{\bfc_i}\hat{A}_\ur^\dagger\hat{A}_{\ur^\prime}\ket{\bfc_j}= (1-\gamma)^{N-r}\gamma^{r}\binom{N}{r}\sqrt{\binom{r}{\ur}\binom{r}{\ur^\prime}}\\
\times\left(\sum_{\un}(\alpha_{\un}^{(i)})^*\alpha^{(j)}_{\un-\ur+\ur^\prime}\frac{\binom{N-r}{\un-\ur}}{\sqrt{\binom{N}{\un}\binom{N}{\un-\ur+\ur^\prime}}}\right).
\end{multline}
The Vandermonde convolution, \cref{lemma:convolution}, gives 
\begin{align}\label{eq:VondConv}
    \binom{N-r}{\un-\ur}=\sum_{\ue^\prime\in\cS_{q,t-r}}\binom{t-r}{\ue^\prime}\binom{N-t}{\un-\ur-\ue^\prime}.
\end{align}
Our plan is to replace $\binom{N-r}{\un-\ur}$ in \cref{eq:InnerProd} with the sum on the right-hand side, but before that, let us 
introduce new variables $\ue,\uf\in \cS_{q,t}$ such that $\ue=\ur+\ue^\prime$ and $\uf=\ur^\prime+\ue^\prime$ (this is always possible). Making these changes, we obtain
\begin{multline*}
\bra{\bfc_i}\hat{A}_\ur^\dagger\hat{A}_{\ur^\prime}\ket{\bfc_j}=(1-\gamma)^{N-r}\gamma^r\binom{N}{r}\sqrt{\binom{r}{\ur}\binom{r}{\ur^\prime}}\\ \times\sum_{\ue^\prime\in\cS_{q,t-r}}\binom{t-r}{\ue^\prime}\Biggl(\sum_\un(\alpha_{\un}^{(i)})^*\alpha^{(j)}_{\un-\ue+\uf}\frac{\binom{N-t}{\un-\ue}}{\sqrt{\binom{N}{\un}\binom{N}{\un-\ue+\uf}}}\Biggr)
\end{multline*}
Therefore, as long as condition (\hyperlink{C3}{C3}) holds, orthogonality part of the KL conditions is satisfied. Following the same sequence of steps, one can show that (\hyperlink{C4}{C4}) suffices for
the non-deformation condition. Finally, conditions (\hyperlink{C1}{C1}), (\hyperlink{C2}{C2}) should be satisfied for the code
basis to be orthonormal.
\end{proof}

Equipped with this result, we give the following definition.
\begin{definition}[Bosonic distance]\label{def: Bosonic distance}
A $K$-dimensional Fock state code \cref{eq:DefFockStateCodes} has {\em bosonic distance} $d_b=t+1$ 
if it corrects $t$-AD errors.
\end{definition}
Given a code on $q$ modes with total excitation $N$, dimension $K$, and bosonic distance $d_b$, below we call it an
$(N,K,q,d_b)$ Fock state code.

\subsection{Equivalence of Fock state codes and PI codes}
Relying on the results established so far, we now will argue that qudit PI codes and multi-mode Fock state codes are equivalent.
This will become apparent once we introduce some notation. For a $q$-tuple $\un\in \cS_{q,N}$, let $\ket{D_\un}$ be 
a Dicke state and let $\ket{\un}_b$ be a Fock state. There is a mapping between them, which we
denote by $f$, that identifies the excitation pattern of each Dicke state with a corresponding Fock-state label:
    \begin{align}\label{eq:mappingf}
       \ket{D_\un}\quad\xLeftrightarrow[\text{$f^{-1}$}]{\text{$f$}} \quad\ket{\un}_b 
    \end{align}
For example, the Dicke state \(\ket{111}\) corresponds to the Fock state \(\ket{030}_b\), where we remind the reader that a Dicke state is labeled by the composition of the qudit basis-state labels participating in the state 
(see Definition \ref{def:composition}).

This mapping can be seen to arise directly from second quantization.
For example, a two-mode Fock space is the state space of particles in one of two possible orbitals, call them \(\psi\) and \(\phi\).
The vacuum Fock state \(|00\rangle\) then corresponds to the zero-particle state, the two single-occupation Fock states are \(|10\rangle = |\psi\rangle\) and \(|01\rangle = |\phi\rangle\), the three double-occupation Fock states are \(|20\rangle = |\psi\psi\rangle\), \(|11\rangle \propto |\psi\phi\rangle + |\phi\psi\rangle\), and \(|02\rangle = |\phi\phi\rangle\), etc.
Identifying \(\psi \leftrightarrow 0\) and \(\phi \leftrightarrow 1\) yields the Dicke states.

Combining this mapping with our derived relations between the error-correction conditions on erasure and amplitude damping yields Fock state codes from PI codes.

\begin{proposition}\label{prop:equivalence}
Applying the mapping $f$ to a $(N,K,q,t+1)$ PI code results in a $(N,K,q,t+1)$ Fock state code.
\end{proposition}
\begin{proof}
    By combining \cref{Theo:PIMainTheorem} and \cref{Theo:BosonMainTheoremAD}.
\end{proof}
The above correspondences between error-correction conditions and basis states allow us to directly construct bosonic codes that can correct AD errors using the large existing literature of qubit and qudit PI codes \cite{ruskaiExchange,ruskai-polatsek,ouyangHigherDimensions,ouyangPI,ouyangQudit,IanEricDihedralIEEE,exoticGates,kubischtaUnitaryt,aydin2023family,gross}.

\section{Spin codes}
\label{sec:spin}

While a \(D\)-dimensional nuclear manifold is typically interpreted as an irrep of its corresponding angular momentum group \(SU(2)\), it also houses irreps \(SU(q)\) for any \(q\) that admits an irrep of dimension \(D\).
This interpretation allows us to construct new codes for hyperfine systems that pack information in a different way than what is dictated by the \(SU(2)\) interpretation.

For example, the discrete simplices \(\cS_{2,5}\) and \(\cS_{3,2}\) both consist of 6 points.
This implies that a six-dimensional space houses both the spin-\(5/2\) irrep of \(SU(2)\) and the \((2,0)\)-irrep of \(SU(3)\)~\cite{georgi2000lie}.
The geometry of the former is a ladder of angular momentum states (a 1-simplex), while the latter states from a triangle~\cite[Eq.~(9.30)]{georgi2000lie} (a 2-simplex).
A trivial code can then be defined for both using the corners of each simplex, with the \(SU(3)\) case yielding a higher logical dimension. 

More powerful codes can be defined similarly using the advantages of packing in higher-dimensional simplices, similar to the advantages of using higher-dimensional spheres for quantum spherical codes~\cite{jain2024quantum}.
Since nuclear manifolds are often completely controllable~\cite{asaad2020coherent,fernandez2024navigating,low2025control}, such codes are no less realizable than their \(SU(2)\) counterparts.

\subsection{\texorpdfstring{\(SU(q)\)}{} and the Jordan-Schwinger map}
\label{subsec:js}
We define spin codes from Fock state codes using the Jordan-Schwinger (JS) mapping (more precisely, its \(SU(q)\) generalization~\cite{klein1991boson}).
This mapping allows us to simply associate each Fock state \(|\underline n\rangle_b\) with the spin state \(|\underline n\rangle_s\) while also relating bosonic noise to noise on this spin subspace.

Let $\{J_i : i=1,2,\ldots,q^2-1\}$ be the generators of the Lie algebra $\mathfrak{s u}(q)$ in the fundamental (i.e., defining) \(q\)-dimensional irreducible representation, or irrep~\cite{georgi2000lie}.
The JS map yields a simple expression for the representation of this Lie algebra in terms of operators that are quadratic in the bosonic creation and annihilation operators.
These, in turn, generate the corresponding Lie group of passive linear-optical transformations acting on \(q\) bosonic modes.

For $i=1,\dots,q$, let 
    \begin{gather*}
    \ah_i=\underbrace{I\otimes\dots\otimes I}_{i-1}\otimes~\ah\otimes
       \underbrace{I\otimes\dots\otimes I}_{q-i}
    \end{gather*}
 be annihilation operators, with their adjoints being the creation operators. The JS quadratic representation of the Lie algebra is then
\begin{align}\label{eq:js_lie-algebra}
    \sum_{j,k=1}^{q}\ah_j^\dagger(J_i)_{jk}\ah_k, \quad i=1,\ldots,q^2-1~.
\end{align}

This Lie algebra can be upgraded to \(\mathfrak{u}(q)\) by plugging in the identity matrix for \(J_i\) to yield the total excitation operator.
This is a constant, \(N\), on our simplex-labeled subspaces.
In other words, rotations generated by quadratic combinations of the above type preserve the total excitation number \(N\).

For each \(N\), Fock states \(|\underline{n}\rangle_b\) are in one-to-one correspondence with complex-valued monomials of the form \(z_1^{n_1}\cdots z_q^{n_q}\).
As such, the representation of \(\mathfrak{u}(q)\) on that fixed-\(N\) space is Sym$^N(\mathbb{C}^q)$, i.e., the irreducible representation on homogeneous degree-\(N\) polynomials in \(q\) variables~\cite{gitman1993coherent}, \cite[Sec.~11.2.2]{VilenkinKlimyk1993}
(\textit{a.k.a.}~the completely symmetric irrep~\cite{klein1991boson}).
The fundamental irrep is present in the single-excitation space, i.e., at \(N=1\).

Our simplex mapping then associates such irreps with irreps acting on Dicke states \(|D_{\underline n}\rangle\), but we can also think of an isolated nuclear manifold as housing an irrep of the appropriate dimension.
This yields the ``spin'' codes that we now define.

The Hilbert space of a spin-$N$ system over $SU(q)$ has the basis $\{\ket{\un}_s : \un\in\cS_{q,N}\}$. 
A $K$-dimensional {\em spin code} $\cQ_{sp}$ is a subspace of the spin-$J$ system with the logical basis
\begin{align}\label{eq:DefSpinCode}
    \ket{\bfc_i} = \sum_{\un\in\cS_{q,N}}\alpha^{(i)}_{\un}\ket{\un}_s, \quad i=0,1,\ldots,K-1,
\end{align}
where $\alpha^{(i)}_\un$ are complex coefficients. 

\begin{definition}[Spin distance]
    A spin code $\cQ_{sp}$ has spin distance $d_s=t+1$ if it detects the errors from the following set:
    \begin{align}\label{eq:spinErrorSet}
        \epsilon_t^s=\{J_{i_1}J_{i_2}\dots J_{i_t} : i_j=1,2,\ldots,q^2-1\}.
    \end{align}
\end{definition}
Given an $SU(q)$ spin code of total spin $N$, dimension $K$, and distance $d_s$, we refer to it as an $(N,K,q,d_s)$ spin code.

\subsection{Error Correction Conditions}
In the following theorem, we show that the error-correction conditions for $SU(q)$ spin codes are related to qudit PI codes.

\begin{theorem}\label{Theo:SpinMainTheorem}
Let $\cQ_{sp}$ be an $(N,K,q,d_s)$ code defined by the basis coefficients $\alpha_\un^{(i)}, \un\in \cS_{q,N}, 0\le i\le K-1$ as in \eqref{eq:DefSpinCode}. If $\cQ_{sp}$ satisfies conditions {\rm (\hyperlink{C1}{C1})--(\hyperlink{C4}{C4})}, then its distance $d_s= t+1$.
\end{theorem}
\begin{proof}
Recall that a spin code has a distance $d_s=t+1$ if it can detect the errors from the set $\epsilon_t^s$ as in \eqref{eq:spinErrorSet}. Let $ \mathbf{J}=J_{i_1}J_{i_2}\ldots J_{i_t}\in\epsilon_t^s $. Then to satisfy the KL orthogonality condition, we need to show that
\begin{align}\label{eq:OrthCondSpinExpanded}
    \bra{\bfc_i}\bfJ\ket{\bfc_j}=\sum_{\un\in\cS_{q,N}}\sum_{\un^\prime\in\cS_{q,N}}(\alpha_{\un}^{(i)})^*\alpha_{\un^\prime}^{(j)}\bra{\un}\bfJ\ket{\un^\prime}_s=0
\end{align}
for all $\bfJ\in\epsilon_t^s$. Define the mapping from an arbitrary spin state on $SU(q)$ to an arbitrary $q$-mod Fock state
\begin{align}\label{eq:smap}
    \ket{\un}_s \quad\xLeftrightarrow[\text{$s^{-1}$}]{\text{$s$}}\quad\ket{\un}_b
\end{align}
By definition of the JS map, we have
\begin{align*}
    s\left(J_i\ket{\un}_s\right) = \operatorname{js}(J_i)s(\ket{\un}_s) = \operatorname{js}(J_i)\ket{\un}_b, \quad i=1,\ldots,q^2-1,
\end{align*}
and therefore,
\begin{align}\label{eq:JSmappedInnerProd}
    \bra{\un}\mathbf{J}\ket{\un^\prime}_s = \bra{\un}\operatorname{js}(\mathbf{J})\ket{\un^\prime}_b.
\end{align}
Note that the JS map of $J_i$ has the form
\begin{align*}
    \operatorname{js}(J_i)=\sum_{j,k}\gamma_{jk}^{(i)}\ah^\dagger_j\ah_k,
\end{align*}
where $\gamma_{jk}^{(i)}\in \mathbb{C}$.
 Hence the image of $\bfJ=J_{i_1}J_{i_2}\ldots J_{i_t}$ under JS has the form 
\begin{align*}
\operatorname{js}(\mathbf{J})=\sum_{j_1,k_1}\cdots\sum_{j_t,k_t}(\gamma_{j_1k_1}^{i_1}\ldots\gamma_{j_tk_t}^{i_t})\ah_{j_1}^\dagger\ah_{k_1}\ldots \ah_{j_t}^\dagger\ah_{k_t}
\end{align*}
For an integer $r\ge0$ and $\ue\in \cS_{q,r}$, define the operator 
\begin{align*}
    \hat{B}_{\ue}= \ah^{e_0}\otimes\ah^{e_1}\otimes\ldots\otimes \ah^{e_{q-1}}.
\end{align*}
Recalling the commutation relations $[\ah_i,\ah_j^\dagger]=\delta_{i,j}$, $[\ah_i,\ah_j]=0$, and $[\ah_i^\dagger,\ah_j^\dagger]=0$, we have
\begin{align*}
    js(\mathbf{J})=\sum_{r=0}^t\sum_{r^\prime=0}^t\sum_{\ur\in\cS_{q,r}}\sum_{\ur^\prime\in \cS_{q,r^\prime}}  \rho_{\ur,\ur^\prime}\hat{B}_{\ur}^\dagger\hat{B}_{\ur^\prime}
\end{align*}
Here $\rho_{\ur,\ur^\prime}$ is a set of complex coefficients. Then, using \eqref{eq:JSmappedInnerProd}, we have
\begin{align*}
    \bra{\un}\mathbf{J}\ket{\un^\prime}_s=\sum_{r=0}^t\sum_{r^\prime=0}^t\sum_{\ur\in\cS_{q,r}}\sum_{\ur^\prime\in \cS_{q,r^\prime}}  \rho_{\ur,\ur^\prime} \bra{\un}\hat{B}_{\ur}^\dagger\hat{B}_{\ur^\prime}\ket{\un^\prime}_b
\end{align*}
For a single mode we have
  $$
  \ah^e\ket n_b=\sqrt{e!\binom ne}\ket{n-e}_b.
  $$
Together with the definition of $\hat B_\ue$ this yields
\begin{align*}
    \hat{B}_\ur\ket{\un}_b&= 
    \sqrt{r_0!\dots r_{q-1}!\binom{n_0}{r_0}\ldots\binom{n_{q-1}}{r_{q-1}}}\ket{\un-\ur}_b\\
    &=\sqrt{\frac{r!\binom{N}{r}\binom{N-r}{\un-\ur}}{\binom{N}{\un}}}\ket{\un-\ur}_b, 
\end{align*}
where the second equality follows by \cref{Lemma: multinomial}.
Then the inner product
\begin{align*}
    \bra{\un}\hat{B}_{\ur}^\dagger\hat{B}_{\ur^\prime}\ket{\un^\prime}_b&=\sqrt{\frac{r!r^\prime!\binom{N}{r}\binom{N}{r^\prime}\binom{N-r}{\un-\ur}\binom{N-r^\prime}{\un^\prime-\ur^\prime}}{\binom{N}{\un}\binom{N}{\un^\prime}}}\delta_{\un^\prime-\ur^\prime,\un-\ur}\\
    &=r!\binom{N}{r}\frac{\binom{N-r}{\un-\ur}}{\sqrt{\binom{N}{\un}\binom{N}{\un+\ur^\prime-\ur}}}\delta_{\un^\prime,\un-\ur+\ur^\prime}
\end{align*}
Repeating the same sequence of steps as in eqs. \eqref{eq:InnerProd},\eqref{eq:VondConv}, we obtain
\begin{multline*}
  \bra{\un^\prime}\hat{B}_{\ur^\prime}^\dagger\hat{B}_{\ur}\ket{\un}_b=\\
 r!\binom{N}{r}\sum_{\ue^\prime\in\cS_{q,t-r}}\binom{t-r}{\ue^\prime}\frac{\binom{N-t}{\un-\ue}}{\sqrt{\binom{N}{\un}\binom{N}{\un-\ue+\uf}}}\delta_{\un^\prime,\un-\ue+\uf} 
\end{multline*}
Here the tuples $\ue,\uf\in\cS_{q,t}$ are defined as $\ue=\ur+\ue^\prime$, $\uf=\ur^\prime+\ue^\prime$
Therefore, we have
\begin{multline}\label{eq:spinInnerProd}
    \bra{\un}\mathbf{J}\ket{\un^\prime}_s=\\
    \sum_{r=0}^t\sum_{\ur,\ur'\in\cS_{q,r}}\sum_{\ue^\prime\in\cS_{q,t-r}}\binom{t-r}{\ue^\prime}\chi_{\ur,\ur^\prime}\frac{\binom{N-t}{\un-\ue}}{\sqrt{\binom{N}{\un}\binom{N}{\un-\ue+\uf}}}\delta_{\un^\prime,\un-\ue+\uf} 
\end{multline}
Here $\chi_{\ur,\ur^\prime}= r!\binom{N}{r}\rho_{\ur,\ur^\prime}$.
Combining \eqref{eq:OrthCondSpinExpanded} with \eqref{eq:spinInnerProd}, we obtain 
\begin{multline*}
    \bra{\bfc_i}\mathbf{J}\ket{\bfc_j}=\\
    \sum_{r=0}^t\sum_{\ur_1\in\cS_{q,r}}\sum_{\ur^\prime\in \cS_{q,r}}\sum_{\ue^\prime\in\cS_{q,t-r}}\binom{t-r}{\ue^\prime}\chi_{\ur,\ur^\prime}\\
    \times\left(\sum_{\un\in\cS_{q,N}}(\alpha^{(i)}_{\un})^*\alpha^{(j)}_{\un-\ue+\uf}\frac{\binom{N-t}{\un-\ue}}{\sqrt{\binom{N}{\un}\binom{N}{\un-\ue+\uf}}}\right),
\end{multline*}
which is zero as long as condition (\hyperlink{C3}{C3}) is satisfied. Finally, as in \cref{Theo:BosonMainTheoremAD}, we can show that the non-deformation condition is also satisfied when condition (\hyperlink{C4}{C4}) is met. 
\end{proof}

\subsection{Equivalence of spin codes and PI codes}
In this subsection, we establish a relationship between $SU(q)$ spin codes and qudit PI codes. To do that we define the following mapping between spin states and PI states:
\begin{align}\label{eq:eqMapSpinPI}
    \ket{D_{\un}}\quad\xLeftrightarrow[\text{$\sigma^{-1}$}]{\text{$\sigma$}} \quad \ket{\un}_s.
\end{align}
Here $\un\in\cS_{q,N}$. This mapping allows us to construct spin codes using the PI codes.
\begin{proposition}\label{prop:equivalenceSpinPI}
Applying the mapping $\sigma$ to a $(N,K,q,t+1)$ PI code results in a $(N,K,q,t+1)$ spin code.
\end{proposition}
\begin{proof}
    By combining \cref{Theo:PIMainTheorem} and \cref{Theo:SpinMainTheorem}.
\end{proof}
Using \cref{prop:equivalenceSpinPI} along with the literature of qubit and qudit PI codes, it is possible to construct spin codes over $SU(q)$ that can correct for random rotation errors. 

\remove{
We briefly mention a connection between PI codes and spin codes. It will play a role in Sec.~\ref{sec: JordanSchinger} where we detail a link between spin codes and Fock state codes and their logical gates. 

Spin systems arise as irreducible representations of the unitary group $SU(2)$. In particular, a spin-$J$ system corresponds
to a $(2J+1)$-dimensional representation. The Hilbert space of a spin-$J$ system has the basis $\ket{J,m}, m=-J,-J+1,\ldots,J-1,J$, where $\ket{J,m}$ is an eigenstate of the $z$ component of the angular momentum operator, $J_z$, with eigenvalue $m$.  A $K$-dimensional {\em spin code} $\cQ_{sp}$ \cite{gross} is a subspace of the spin-$J$ system with the logical basis
\begin{align}\label{eq:DefSpinCode}
    \ket{\bfc_i} = \sum_{m}\alpha^{(i)}_{m}\ket{J,m}, \quad i=0,1,\ldots,K-1,
\end{align}
where $\alpha^{(i)}$ are complex coefficients. Spin codes are designed to protect
the encoded information against small-order isotropic rotations. 
\begin{definition}[Spin distance]
    A spin code $\cQ_{sp}$ as in \cref{eq:DefSpinCode} has spin distance $d_s=t+1$ if it can detect the errors from the following set
    \begin{align*}
        \cE^s_t=\{ \hat{E}^r_{\delta m} : |\delta m|\leq r \leq t \},
    \end{align*}
    where
    \begin{align*}
        \hat{E}^r_{\delta m}\propto \sum_{m=_J}^J C^{J,m+\delta m}_{J,m;r,\delta m}\ket{J,m+\delta m}\bra{J,m},
    \end{align*}
and where $C^{J,m+\delta m}_{J,m;r,\delta m}$ are the Clebsch-Gordon coefficients.
\end{definition}
Thus, a spin code with distance $d_s=2t+1$ can detect random rotations of order $2t$.
}

For $q=2$, it is shown in \cite{IanEricDihedralIEEE} that $SU(2)$ spin codes can also be mapped to qubit PI codes through the mapping $\sigma^{-1}$, a mapping that its authors called {\em Dicke bootstrap}. The next lemma shows that $\sigma$ is in fact an isometry.
\begin{lemma}{\rm \cite[Lemma~2]{IanEricDihedralIEEE}}\label{lemma:ericandian}
 The mapping $\sigma^{-1}$ transforms a $(N,K,2,t+1)$ spin code
 to a $(N,K,2,t+1)$ PI code.
 \end{lemma}

\section{From \texorpdfstring{$\ell_1$}{W} to \texorpdfstring{$SU(q)$}{} codes}\label{sec:existence}

In this section we develop an approach to the construction of the three related families of $SU(q)$ 
codes, the qudit PI codes, spin codes, and Fock state codes. 
The main tool involved in the derivations is the connection between PI qudit codes and classical $\ell_1$ codes.
Once we construct a family of PI codes, we employ the connection between them and spin and Fock state codes
to obtain new code families of each of these two types.

In \cref{sec: classical} we overview constructions of $\ell_1$ codes in the simplex. In \cref{sec: l1-PI-Fock}
we state and prove the main equivalence result between $\cS_{q,N}$ codes and qudit PI codes, which yields a
bound on the parameters of the PI codes. 
In \cref{sec: asymptotic} we analyze the asymptotic scaling of the parameters
qudit PI codes arising from the suggested approach. We show that 
there exist families of $(N,K,N,d)$ codes with increasing $N$
and any $K,d$ satisfying $K =  o(2^{N})$ and $d = o(\frac{N}{\log N})$, where this claim applies to 
$SU(q)$ codes of all the three types, and $d$ is $d_b,d_s$ or the qudit PI code distance as appropriate.

Overall, our construction generalizes the known results \cite{chuangADCode}, \cite{VanLoock}, \cite{OuyangADCode,movassagh2024constructing} both in terms of allowing a range of code dimensions and a better scaling of the distance.

\subsection{Classical \texorpdfstring{$\ell_1$}{l1} codes} \label{sec: classical}

The construction of qudit PI codes presented in the next subsection starts with codes in the $\ell_1$ metric.
The general problem of constructing such codes is well researched in classical coding theory
\cite{varshamov1973class}, \cite{tallini2011}, \cite{kovacevic2018multisets}.
Existing constructions rely in a large part on the Bose-Chowla theorem and related concepts from additive number theory (Sidon sets, $B_h$-sequences), with applications to flash memory, DNA coding, and others \cite{barg2010codes,goyal2024gilbert}. 

The specific setting that we will use in the next section relates to codes in $\cS_{q,N}$ under the $\ell_1$ distance. A distinguishing feature of this code subclass is that the underlying space is formed of all $q$-tuples with {\em fixed} $\ell_1$ norm\footnote{A subclass of these codes where not only the $\ell_1$ norm, but also the composition $\comp(\bfx)$ of every codeword is fixed, plays a major role in information theory \cite{csiszar2011information}.}, viz. \eqref{eq:SimplexDef} and \cref{fig: SqN,fig: S44}. 
Define a distance $d_1$ on $\cS_{q,N}$ by setting 
  \begin{equation}\label{eq: d1}
  d_1(\underline{x},\underline{y}):=\frac{1}{2}\sum_{i=0}^{q-1}|x_i-y_i|
  \end{equation}
 and for a subset $\cC_{\ell_1} \subset \cS_{q,N}$ put
$
    d_1(\cC_{\ell_1})= \min_{\substack{\underline{x},\underline{y}\in \cC_{\ell_1}\\ \underline{x}\neq \underline{y}}}d_1(\underline{x},\underline{y}).
$
We call this distance the $\ell_1$ norm despite the 1/2 factor, which is inserted simply because for constant sum $|\ux|=N$, the true $\ell_1$ distance is always even\footnote{This follows because $\sum_{i=0}^{q-1}(x_i-y_i)=0$, and so $\sum_{i:x_i>y_i}(x_i-y_i)=\sum_{i:x_i<y_i}(y_i-x_i)$.}. Below, we call subsets of the discrete simplex 
{\em $\ell_1$-codes} if their distance is measured with respect to the metric $d_1$.

We present two constructions of $\ell_1$ codes, starting with a simple approach to obtain $\ell_1$ codes with large distance in the simplex $S_{q,N}$ with $N=q$. This assumption will play a role below when we use the constructed codes as a device to obtain Fock state codes. To justify it, note that for constant-excitation Fock state codes, the fidelity as given by \cref{eq: fidelity} does not
depend on the number of modes, $q$, and thus the constructed codes are optimized if we manage to minimize $N$ for a given $t$.

Furthermore, the codes that we construct have a particular size, which fits the requirements of the Fock state code construction.

\begin{theorem}\label{theorem:existence l1}
    Let $K,t\geq 2$ be integers and let $$N=q=(K-1)t(t+1).$$ There is an $\ell_1$ code $\cC_{\ell_1}\subset \cS_{N,N}$ of size
      $$
      |\cC_{\ell_1}|\ge (K-1)|\cS_{N,t}|+1
      $$
and distance $t+1$.    
\end{theorem}
\begin{proof}
Consider the set $Q=(t+1)\cS_{q,(K-1)t}$ formed of all tuples in the simplex scaled by the factor $(t+1)$.
Clearly, $Q\subset\cS_{N,N}$ and 
    $$
    d_1(Q)=(t+1) d_1(\cS_{q,(K-1)t})=t+1>t.
    $$
Now put
   $
        \cC_{\ell_1} = Q\cup\mathbf{1}^q,
$
where $\mathbf{1}^q$ is the all-ones vector. It can be shown that $d_1(\mathbf{1}^q,\ux)>t$ for all $\ux\in \cQ$.
Then, because the number of non-zero elements of $\bfc\in\cQ$ is less than or equal to $(K-1)t$, we have
    \begin{align*}
        \frac{1}{2}|\bfc-\mathbf{1}^q|_1\geq \frac{N-(K-1)t}{2} = \frac{(K-1)t^2}{2},
    \end{align*}
    which is always greater than $t$ as long as $K,t\geq 2$, and thus, $d_1(\cC_{\ell_1})>t$. 
Further,
    \begin{align*}
        |\cC_{\ell_1}| &= \binom{q+(K-1)t-1}{(K-1)t} + 1\\
        &= \binom{(K-1)t^2+2(K-1)t-1}{(K-1)t} + 1\\
        &= \sum_{i=0}^{(K-2)t}\binom{(K-1)t^2+Kt-1}{(K-1)t-i}\binom{(K-2)t}{i} + 1,
     \end{align*}
where on the last line we used Vandermonde convolution. Next observe that for $K,t\ge 2$,
    \begin{align*}
   (K-1)t^2+Kt-1-2Kt+2t=((K-1)t-K)t+2t-1>0 ,
   \end{align*}
implying that $\binom{(K-1)t^2+Kt-1}{(K-1)t-i}$ is a decreasing function of $i$ as it increases in the range $0\le i\le (K-2)t$.
Therefore, for all such $i$,
   $$
   \binom{(K-1)t^2+Kt-1}{(K-1)t-i}\ge \binom{(K-1)t^2+Kt-1}{t}
   $$
and there are $(K-2)t+1\ge K-1$ terms in the sum. Noting that $(K-1)t^2+Kt-1=(K-1)t(t+1)+t-1$, we obtain 
     \begin{align*}   
        |\cC_{\ell_1}|&\geq (K-1)\binom{(K-1)t(t+1)+t-1}{t} + 1,
    \end{align*}
which is precisely the claimed bound.
\end{proof}
\begin{remark}\label{remark: count}
    Our choice of the bound for the code size $|\cC_{\ell_1}|$ may seem arbitrary: indeed, the last inequality can be tightened with little effort. The reason for this choice is related to the count of coefficients in the error correction conditions
    (\hyperlink{C1}{C1})--(\hyperlink{C4}{C4}). This link will become apparent in the proof of \cref{prop:bound}
below. \end{remark}

Next, we mention another way of constructing $\ell_1$ codes in the simplex, based on a classic approach from the literature.
\begin{definition}
    Let $G$ be an Abelian group, written additively. A subset $B\subseteq G$ is a $t$-{\em Sidon set} if the sums
    \begin{align*}
        b_{i_1}+b_{i_2}+\ldots+b_{i_t}
    \end{align*}
    are all distinct for $0\leq i_1\leq i_2 \leq \ldots \leq i_{t}\leq |G|$. 
\end{definition}
A well-known way of constructing Sidon sets is provided by the Bose-Chowla theorem, which we cite in Sec.~\ref{sec: asymptotic} below.

The next result appears in many places in the literature; see \cite{barg2010codes} and its references. We cite it in the form given in \cite[Thm.~14]{kovacevic2018multisets}.
\begin{theorem}\label{theorem:Kovacevic}
     Let $G$ be an Abelian group that contains a $t$-Sidon set of cardinality $q$. Then, for every $q\geq2$ and $n>t\geq 1$, there exists an $\ell_1$ code $\cC_{\ell_1}\subset \cS_{q,N}$ with distance $\ge t+1$ such that
    \begin{align}
        |\cC_{\ell_1}| \geq \frac{|\cS_{q,N}|}{|G|}.   \label{eq: LowerBounBoseChowla}
    \end{align}
\end{theorem}
The construction of the code $\cC_{\ell_1}$ relies on the following greedy argument: given a $t$-Sidon set $\{g_1,\dots,g_q\}\subset G$,
consider the set of vectors $$\cC_{\ell_1,g}:=\Big\{\underline x\in \cS_{q,N} \mid\sum_{i=0}^{q-1} x_i g_i=g\Big\},$$ where $g\in G$ is some group element. Plainly, the code distance is $\ge t+1$, because the opposite inequality would violate the Sidon set condition. Further, the codes
$\cC_{\ell_1,g}, g\in G$ form a partition of $\cS_{q,N}$, so at least one of them is of size that satisfies \cref{eq: LowerBounBoseChowla}.

As a result of this argument, to construct codes of large size we therefore need small groups that host a $t$-Sidon set. This will be addressed in \cref{sec: asymptotic}, where we construct Fock state codes (as well as PI codes) with better parameter scaling than the known results. Our construction relies on the group $\Z_m$  and $t$-Sidon sets inside it, where
$m$ depends on $q$ and $t$.

\begin{remark}
Yet another way to prove the existence of $\ell_1$ codes with large distance is given by a greedy procedure commonly called 
the Gilbert--Varshamov bound (it does not yield explicit codes). In particular, it is known that there exist codes in the simplex $\cS_{q,N}$ of size at least $|\cS_{q,N}|/\overline{V}(t)$ and distance $t+1$, where $\overline{V}(t)$ is the {\em average} volume of the $\ell_1$-ball of radius $t$ in the space \cite{gu1993generalized,tolhuizen1997generalized}. A recent work \cite{goyal2024gilbert} performed an asymptotic analysis of this bound for large $N$ and $t=\tau N, \tau>0$. 
Although their results will likely yield the existence of Fock state codes whose distance scales linearly with $N$, they involve complicated expressions with not much insight into the properties of the codes, so we do not cite them here. 
\end{remark}

\subsection{From \texorpdfstring{$\ell_1$}{l1-2} to \texorpdfstring{$SU(q)$}{SUq}: PI, Fock state, and spin codes}\label{sec: l1-PI-Fock}
In this section, we introduce the framework for the code construction relying on the error correction conditions derived earlier for $SU(q)$ codes (\cref{fig: KL conditions}).
The underlying idea of
is to show that conditions (\hyperlink{C1}{C1})--(\hyperlink{C4}{C4}) translate into a set of linear equations for the basis coefficients $\alpha_{\un}^{(i)}$, \cref{eq: basis vectors,eq:DefFockStateCodes,eq:DefSpinCode}, where the coefficients that define the basis states are indexed by subsets of vectors of a classical code in $\cS_{q,N}$
that corrects errors in the $\ell_1$ metric. To identify those subsets, we rely on a 
classic result from convex geometry.

\remove{In the next proposition, we rephrase this lemma to match our context. For completeness, we provide an independent short proof.
\begin{proposition}\label{prop: Radon}
    Let $X=\{x_1,x_2,\ldots,x_n\}$ be a set of variables. Let $E=\{e_1,e_2,\ldots,e_m\}$. Define the real numbers $\alpha_{e,x}$, where $e\in E$ and $x\in X$. If $n\geq m+2$, then there exists a partition of $X$ into two nonempty disjoint subsets $X_1,X_2$ such that the set of equations with respect to $X$
    \begin{align*}
        \sum_{x\in X_1}x &= \sum_{x\in X_2}x\\
        \sum_{x\in X_1}xa_{e,x} &= \sum_{x\in X_2}xa_{e,x} \quad \text{for all } e\in E
    \end{align*}
    has a nonzero solution with $x_i\geq 0$ for $i\in\{1,2,
    \ldots,n\}$
\end{proposition}
\begin{proof} For $i=1,2,\dots, n$, let $A_i=(a_{e,x}, e\in E)^\intercal\in \R^m.$ Consider the set of linear equations
   in   variables $x_1,\dots,x_{n}$:
    \begin{gather*}
    \sum_i x_iA_i=0,\quad \sum_i x_i=0.
    \end{gather*}
By dimension counting, there is a nonzero solution $\bfx$. Define two disjoint subsets, $P_1=\{i:x_i>0\}$ and $P_2=\{i:x_i<0\}$ and let $A^{(1)}=\{A_i:i\in P_1\}$ and $A^{(2)}=\{A_i:i\in P_2\}$.
We obtain
   $$
   \sum_{i\in P_1}x_iA_i=\sum_{i\in P_2}(-x_i)A_i, \quad \sum_{i\in P_1}x_i=\sum_{i\in P_2}(-x_i)
   $$ 
as claimed.\footnote{While we did not need the ``comvex combination'' part of Radon's lemma, getting there
amounts to normalizing the solution by $\beta:=\sum_{i\in P_1}x_i$.}
\end{proof}

\begin{lemma}[Radon's lemma, \cite{matousek2013lectures}]\label{Radon's Lemma}
    Let A be a set of at least $m+2$ points in $\mathbb{R}^m$. Then there exist two disjoint subsets $A^{(1)},A^{(2)}\subseteq A$ such that
    \begin{align*}
        \operatorname{conv}(A^{(1)})\cap\operatorname{conv}(A^{(2)})\neq \emptyset
    \end{align*}
\end{lemma}}

Recall that the convex combination of points $A_1,A_2,\ldots,A_n\in\mathbb{R}^m$ is defined as
\begin{align*}
    \operatorname{conv}(A_1,A_2,\ldots,A_n) = \Big\{\sum_{i=1}^n \beta_iA_i \mid \beta_i\geq 0, \sum_{i}\beta_i=1 \Big\}.
\end{align*}
The following statement is known as Tverberg's theorem,  \cite{tverberg1966generalization}, or see \cite[p.200]{matousek2013lectures}, \cite{barany2018tverberg}.

\begin{theorem}[Tverberg's theorem]\label{theo:Tverberg}
     Let A be a set of at least $(m+1)(K-1)+1$ points in $\mathbb{R}^m$. Then there exist $K$ pairwise disjoint subsets $A^{(1)},A^{(2)},\ldots,A^{(K)}\subset A$ such that
    \begin{align*}
        \operatorname{conv}(A^{(1)})\cap\operatorname{conv}(A^{(2)})\cap\ldots\cap \operatorname{conv}(A^{(K)})\neq \emptyset
    \end{align*}
\end{theorem}

Rephrased for our needs, this theorem implies the following.
\begin{proposition}\label{prop: tverberg}
Let $\cX=\{x_1,x_2,\ldots,x_n\}$ be a set of variables. Let $E=\{e_1,e_2,\ldots,e_m\}$, and let
$(a_{e,x})_{e\in E,x\in \cX}$ be a set of real numbers. If $n\geq (m+1)(K-1)+1$, then there exists a partition 
of the set $\{1,2,\dots,n\}$, 
   \begin{equation}\label{eq: TP}
   {\mathscr I}_{n,K}=\bigsqcup_{j=1}^K I_j,
   \end{equation}
into a disjoint union such that the set of equations
   \begin{equation}\label{eq: K system}
    \begin{aligned}
        &\sum_{i\in I_1}x_i = \sum_{i\in I_2}x_i = \dots = \sum_{i\in I_K}x_i\\
        &\sum_{i\in I_1}x_ia_{e,i} = \sum_{i\in I_2}x_ia_{e,i}=\dots=\sum_{i\in I_K}x_ia_{e,i} \quad \forall e\in E
    \end{aligned}
   \end{equation}
    has a nontrivial solution $(x_1,\dots,x_n)\in \R^n_{\ge 0}$. 
 \end{proposition}
\begin{proof}
For $i=1,2,\dots, n$, let $A_i=(a_{e,i}, e\in E)^\intercal\in \R^m$
and let $A:=\{A_1,A_2,\ldots,A_n\}$ be the set of $n$ points in Tverberg's theorem. There exist $K$ 
pairwise disjoint subsets $A^{(1)},\ldots,A^{(K)}\subset A$ such that the intersection of their 
convex hulls is not empty. Let $P\in \R^m$ be a point in this intersection and let $I_j: = \{i: A_i\in A^{(j)}\}, \quad j=1,2,\ldots,K$. There
exists a vector of coefficients $\bfx=(x_1,\dots,x_n)\ge 0$ such that
  $$
 \sum_{i\in I_j} x_i A_i=P, \quad\sum_{i\in I_j}x_i=1 \quad\text{for all }j=1,\dots,K.
   $$
Once we put $X_j:=\{x_i: i\in I_r\}$ for all $j$, this gives the desired nonnegative
solution of the system \eqref{eq: K system}.
\end{proof}
The condition $\sum_{i\in I_j}x_i=1$ is redundant for our needs, but equations \eqref{eq: K system} are homogeneous, and the vector of solutions $\bfx$ can be scaled by any positive factor.
Any partition ${\mathscr I}_{n,K}$ of the form \eqref{eq: TP} is called a {\em Tverberg partition} below.

The next statement forms the main technical result of our work, which ties together the code families and auxiliary results introduced above. The point that it makes that once we manage to solve, in any way, 
equations (\hyperlink{C1}{C1})--(\hyperlink{C4}{C4}) for the basis coefficients, we obtain
$SU(q)$ codes of all the three types that we consider.
 \begin{theorem}\label{prop:bound} Let $\cQ$ be one of {\rm \{PI code, spin code, Fock state code\}}.
If there exists an $\ell_1$ code $\cC_{\ell_1}\subset \cS_{q,N}$ with distance $d_1(\cC_{\ell_1})\ge t+1$ such that
     \begin{align}\label{eq: LowerBound}
         |\cC_{\ell_1}|\geq (K-1)\binom{q+t-1}{q-1} + 1,
     \end{align}
then there exists a code $\cQ$ with parameters $(N,K,q,t+1)$.
 \end{theorem}
\begin{proof} 
The proof involves $\ell_1$ codes, used to construct three different types of $SU(q)$ codes.
Our goal is to show that, for the quantum codes that we are constructing, conditions (\hyperlink{C1}{C1})--(\hyperlink{C4}{C4})
are satisfied, implying the distance bound. Let us fix an $\ell_1$ code denoted by $B$ below in the proof. The way we construct $SU(q)$ codes will rely on a partition of the code $B$ into $K$ subsets. To satisfy conditions (\hyperlink{C1}{C1}), (\hyperlink{C3}{C3}), any partition will suffice, while 
conditions (\hyperlink{C2}{C2}), (\hyperlink{C4}{C4}) will rely on a Tverberg partition. Therefore, we will fix a Tverberg
partition from the outset, but first we need to match the parameters in the statement to \cref{prop: tverberg}. 

Let $\{1,2,\dots,B\}$ be the set of indices of the variables $x_i$ and let $\cS_{q,t}$ play the role of $E$ above. Consider the set of 
points $$\Big\{a_{\uh}=(a_{\ue,\uh})\in \reals^{|\cS_{q,t}|}, \uh\in B\Big\}$$ with coordinates
  \begin{equation}\label{eq: coefficients}
   a_{\ue,\uh}={\binom{N-t}{\uh-\ue}}\Big/{\binom{N}{\uh}}, \quad  \end{equation}
and let ${\mathscr I}_{|B|,K}$ be a corresponding Tverberg partition. Note that the points are
indexed by tuples (codewords) in $B$, so we will denote blocks of this partition by $B_i, i=0,1,\dots, K-1$. 

The following argument, phrased for PI codes, applies equally to Fock state and spin codes by \cref{prop:equivalence} and \cref{prop:equivalenceSpinPI}.
Consider a code $\cQ_{PI}$ defined by the basis
    \begin{align*}
        \ket{\bfc_i}=\sum_{\un\in\cS_{q,N}}\alpha^{(i)}_\un\ket{D_\un},\quad i=0,1,\dots,K-1,
    \end{align*}
where
    \begin{align*}
        \alpha_\un^{(i)} \remove{:= \1_{(\un\in B_i)}\mu_\un}:=\sum_{\uh\in B_i}\mu_\uh\delta_{\un,\uh},\quad i\in\{0,1,\ldots,K-1\}.
    \end{align*}
    
The parameters $\mu_\uh$ on the right are as yet undefined; below they will be related to the unknowns $x_\uh$ in \cref{eq: K system}. We will show that
it is possible to choose the set $(\mu_\uh)_\uh$ to satisfy the error correction conditions.
For the moment, we will argue that no matter what $\mu_\uh$ are, Conditions (\hyperlink{C1}{C1}), (\hyperlink{C3}{C3})
are satisfied with our definition of $\alpha_{\un}^{(i)}$.

First, let us verify (\hyperlink{C1}{C1}). It is trivially satisfied since
   $$
     (\alpha_\un^{(i)})^*\alpha_\un^{(j)}=\sum_{\uh\in B_i}\sum_{\uh'\in B_j}
    \mu_\uh^* \mu_{\uh'}\delta_{\un,\uh}\delta_{\uh,\uh'}=0,
   $$
where we used the fact $\delta_{\un,\uh}\delta_{\un,\uh^\prime}=\delta_{\un,\uh}\delta_{\uh,\uh^\prime}$ and $\delta_{\uh,\uh^\prime}=0$. Let us address (\hyperlink{C3}{C3}). Let $\ue,\uf\in \cS_{q,t}$ and observe that for any two distinct
tuples $\uh,\uh'\in B$, $h\ne h'+e-f$ by the assumption $d_1(B)\ge t+1$. Then note that    
      $$
    \delta_{\un,\uh}\delta_{\un-\ue+\uf,\uh^\prime}=\delta_{\un,\uh}\delta_{\uh,\uh^\prime+\ue-\uf}=0
      $$
and 
    \begin{align*}
        (\alpha^{(i)}_{\un})^*\alpha^{(j)}_{\un-\ue+\uf}=\sum_{\uh\in B_i}\sum_{\uh^\prime\in B_j}\mu_\uh^*\mu_{\uh^\prime}\delta_{\un,\uh}\delta_{\uh,\uh^\prime+\ue-\uf}=0
    \end{align*}
yielding (\hyperlink{C3}{C3}). 

Now we will show that there is a choice of the coefficients $\mu_\uh$ that makes Conditions (\hyperlink{C2}{C2}), (\hyperlink{C4}{C4}) turn into equalities. Let us start with rewriting (\hyperlink{C4}{C4}). First, observe that with $\ue,\uf,\uh,\uh'$ as above, $\delta_{\uh,\uh^\prime+\ue-\uf}\ne0$ only if $\ue=\uf$ and $\uh=\uh^\prime$. Therefore, 
    \begin{align*}
        (\alpha^{(i)}_{\un})^*\alpha^{(i)}_{\un-\ue+\uf}&=\sum_{\uh\in B_i}\sum_{\uh^\prime\in B_i}\mu_\uh^*\mu_{\uh^\prime}\delta_{\un,\uh} \delta_{\uh,\uh^\prime+\ue-\uf}\\
        &=\sum_{\uh\in B_i}|\mu_\uh|^2\delta_{\un,\uh}
    \end{align*}
Now put $x_\uh:=|\mu_\uh|^2\ge0$. Then, condition (\hyperlink{C4}{C4}) turns into the set of relations
\begin{align}\label{eq:conditionC4simplified1}
        \sum_{\uh\in B_0}x_\uh\frac{\binom{N-t}{\uh-\ue}}{\binom{N}{\uh}}=\ldots=\sum_{\uh\in B_{K-1}}x_\uh\frac{\binom{N-t}{\uh-\ue}}{\binom{N}{\uh}} \hspace*{.5em} \forall\ue\in\cS_{q,t}.
\end{align}

Let us address the last remaining condition, (\hyperlink{C2}{C2}). First, we rewrite it for our set of coefficients. Following the same steps as above, we find
    \begin{align*}
        |\alpha_{\un}^{(i)}|^2=\sum_{\uh\in B_i}|\mu_\uh|^2\delta_{\un,\uh},\quad i\in\{0,1,\ldots,K-1\}
    \end{align*}
and thus, Condition (\hyperlink{C2}{C2}) turns into
   \begin{align}\label{eq:conditionC2simplified}
        \sum_{\uh\in B_0}x_\uh=\sum_{\uh\in B_1}x_\uh=\ldots=\sum_{\uh\in B_{K-1}}x_\uh.
    \end{align}
Having written our error correction conditions \eqref{eq:conditionC4simplified1}, \eqref{eq:conditionC2simplified} in the form of \cref{eq: K system}, we deduce from Tverberg's theorem that this system has a nonnegative solution if 
$|B|\geq (|\cS_{q,t}|+1)(K-1)+1$ (at this point, the reader may recall our \cref{remark: count}).

Our assumption in \cref{eq: LowerBound} is weaker than this inequality. This relaxation is possible because our 
specific point set $(a_\uh)_\uh$, \cref{eq: coefficients}, in fact lives in a lower-dimensional subspace. Indeed, let $l=(l_\ue)_{\ue\in \cS_{q,t}},$ where $l_{\ue}=\binom t \ue$ for all $\ue$. Then $l^\intercal\cdot a_\uh=1$ for all $\uh\in B$ by \cref{lemma:convolution}. Thus, this point set is contained in an $(|\cS_{q,t}|-1)$-dimensional affine hyperplane, and Tverberg's theorem implies that $(K-1)|\cS_{q,t}|+1$ points suffice for a nonnegative solution to the error-correction conditions \eqref{eq:conditionC4simplified1}, \eqref{eq:conditionC2simplified}. This is precisely our claimed bound \eqref{eq: LowerBound}.

Even though this is not formally needed, it is easy to derive \cref{eq:conditionC2simplified} explicitly. Indeed,
from \cref{eq:conditionC4simplified1}, for any $i,j\in\{0,1,\ldots,K-1\}$, 
   $$
   \sum_{\uh\in B_i} a_\uh x_\uh=\sum_{\uh\in B_j} a_{\uh}x_\uh.
$$
Multiplying by $l$ on both sides and taking $l$ inside the sums by linearity, we obtain an equality in \cref{eq:conditionC2simplified}.

Let $(x_i, i=1,\dots,|B|)$ be a nonnegative solution of \cref{eq:conditionC4simplified1}. We have shown that the PI code 
with the basis
    \begin{equation}\label{eq: PI Tverberg}
    \ket{\bfc_i}=\sum_{\uh\in B_i}\sqrt{x_{\uh}}\ket{D_\uh}, \quad i=0,1,\dots,K-1
    \end{equation}
 has the parameters as in the statement of the theorem. By \cref{Theo:PIMainTheorem}, this shows our
 claim for PI codes, and \cref{Theo:BosonMainTheoremAD,Theo:SpinMainTheorem} imply it for the
 two remaining code families.
\end{proof}

\begin{remark}\label{remark: Radon} Elements of the idea of constructing codes employed in this section have earlier appeared in \cite{movassagh2024constructing}. To explain their result, recall that the version of Tverberg's theorem for $K=2$ is known as {\em Radon's lemma} \cite{matousek2013lectures}, which says that $m+2$ or more points in $\R^m$ can be partitioned into two subsets whose convex hulls have a nonempty intersection. In other words,
given a set $A=\{A_1,A_2,\dots,A_{m+2}\}\subset\R^m$, the system
   \begin{align*}
   \sum_{i\in A^{(1)}} x_i a_{j,i}+\sum_{i\in A^{(2)}} x_i a_{j,i}&=0, \quad j=1,\dots,m \\
   \sum_{i\in A^{(1)}} x_i+\sum_{i\in A^{(2)}}x_i&=0 
   \end{align*}
has a solution satisfying $A^{(1)}\sqcup A^{(2)}=A$ and $x_i\ge 0, i\in A^{(1)};x_i<0, i\in A^{(2)}$. This system is clearly equivalent to \cref{eq: K system}\footnote{To prove the lemma, observe that there are more $x_i$'s than equations, so this system has a nonzero solution, and the partition is naturally formed by the indices of the positive and negative $x_i$. The proof of Tverberg's theorem is much more involved.}.
The authors of \cite{movassagh2024constructing} essentially rediscovered Radon's lemma, designing their code construction algorithm based on it. Our enhanced formalism involving PI codes enabled us to design much more general constructions with better parameter estimates. 

In a related work, \cite{OuyangADCode}, the authors proved the existence of Fock state codes of dimension $K=2$ under the assumption for $\ell_1$ codes of the form 
  $$
   |\cC_{\ell_1}|\geq \sum_{i=0}^t\binom{q+i-1}{q-1} - \binom{t}{2}.
  $$
Our result relies on a less stringent requirement, \cref{eq: LowerBound}, representing an improvement over \cite{OuyangADCode}.
\end{remark}

\begin{remark} \label{remark: EarlierWork} 
As previously mentioned, a lower bound for the size of $K$-dimensional Fock state codes appears in an early work, \cite{chuangADCode}. 
Their bound is weaker than the bound in \cref{eq: LowerBound} in the sense that it relies on larger-size $\ell_1$ codes, which yield Fock-state codes whose codewords are more difficult to realize.
In addition, the argument in that work does not account for the requirement
of a positive solution to the equations for the error-correcting conditions, and therefore appears incomplete.
\end{remark}

\subsection{PI, Fock-state, and spin code\\ 
asymptotics from Sidon sets} \label{sec: asymptotic}

In this section, we analyze sequences of $SU(q)$ codes obtained from the proposed construction. 
Using the connection between these codes and $\ell_1$ codes established above, we begin by finding a sufficiently large-size
$\ell_1$ code with large distance. This will follow by \cref{theorem:Kovacevic} once we bring in the following classic result.
\begin{theorem}[Bose--Chowla, \cite{Chowla1962}]\label{theorem:Bose-Chowla}
     Let $p$ be a prime, $q=p^r$, and $m=(q^{t+1}-1)/(q-1)$. Then, there exist $q+1$ integers, all less than $m$,
    \begin{align*}
        d_0=0, d_1=1, d_2,\ldots, d_q
    \end{align*}
    such that the sums
    \begin{align*}
        d_{i_1}+d_{i_2}+\ldots+d_{i_t}
    \end{align*}
    with $0\leq {i_1}\leq {i_2}\leq \ldots\leq {i_t}\leq q$ are all distinct modulo $m$.
\end{theorem}

Using this theorem, we obtain sequences of $\ell_1$ codes with the parameters as given next.
\begin{proposition}\label{prop:asymptotics}
Consider $\ell_1$ codes in $\cS_{N,N}$. As long as 
  \begin{equation}\label{eq: ll}
  t(1+\log N)+\log K-1\le N,
  \end{equation}
there exists an $\ell_1$ code $|\cC_{\ell_1}|\subset\cS_{N,N}$ with distance $d_1(\cC_{\ell_1})\ge t+1$ and size that satisfies the bound 
    \eqref{eq: LowerBound}.
In particular, if $N\to\infty$ and $t\log N+\log K = o(N)$, there exists a sequence of \ $\ell_1$ codes
that support the conclusion of \cref{prop:bound}.
\end{proposition}
\begin{proof} \cref{theorem:Bose-Chowla} says that the group $\Z_m$ contains Sidon sets of size $q$. 
    By \cref{theorem:Kovacevic} with $G=\Z_m$, there is an $\ell_1$ code $\cC_{\ell_1}$ with $d_1(\cC_{\ell_1})\ge t+1$ and size 
    \begin{align*}
        |\cC_{\ell_1}| \geq \frac{|\cS_{q,N}|}{m}.
    \end{align*}
As remarked above (and also implied by the statement), we assume that $q=N$. We will determine for which $K,N,t$ the right-hand side exceeds the number of points needed for a Tverberg
    partition to exist; cf.~\cref{eq: LowerBound}.
We have 
   $$
   \frac{|\cS_{N,N}|}m=\frac{\binom{2N-1}{N-1}}{m}\ge\frac{\binom {2N}{N}}{2(N^t-(1/N))}\ge\frac{4^{N-1}}{N^{t+1}}
   $$
(using $\binom {2N}{N}\ge 4^N/(2\sqrt N)$). At the same time,
  $$
  (K-1)\binom{N+t-1}{N-1}\le K 2^{N+t}.
  $$
The quotient of these two estimates is
  \begin{align*}
      \frac{4^{N-1}}{K N^t 2^{N+t}}=2^{N-t(1+\log N)-\log K-1}.
  \end{align*}
If the exponent in this expression satisfies \eqref{eq: ll}, there exists a code $\cC_{\ell_1}$ with the stated properties. This proves the first part of our claim.
 Taking $N\to\infty,$ we conclude that, as long as $$\limsup \frac {t\log N+\log K}N<1,$$  there exists a sequence of $\ell_1$ codes
as stated in the proposition (where the above condition is slightly weakened for better readability).
 \end{proof}

Turning to asymptotics, we note that this proposition gives rise to several possible choices of the scaling of $t$ and $K$ that can be used to construct PI codes, Fock state, and spin codes. For instance, the following result is immediately true.
\begin{theorem}\label{theorem:asymptotic}  Let $\cQ$ be one of {\rm \{PI code, spin code, Fock state code\}}.
For $N\to\infty$ and any $K, d$ that satisfy
  $$
  K =o(2^N) \text{ and } d =o\left(\frac N{\log N}\right),
  $$
there exists a sequence of \(\cQ(N,K,q=N,d)\) codes, where $d$ is $d_s,d_b,$ or the quantum PI code distance as appropriate.
\end{theorem}
Constructions of Fock state codes for AD errors were previously studied in \cite{chuangADCode}, \cite{VanLoock}, \cite{OuyangADCode,movassagh2024constructing}. An early work by Chuang et al.~\cite{chuangADCode} presented several examples of codes and claimed the existence of a code family with distance $d_b\propto N^{1/3}$ with an incomplete proof.  The work of Bergmann and van Loock \cite{VanLoock} constructed Fock state codes with distance $d_b\propto \sqrt N$. Finally, Ouyang \cite{ouyangQudit}
gave a construction of qudit PI codes with $d\propto \sqrt{N}$.

\section{Explicit Codes}\label{sec: Examples}

We define several explicit codes as special cases of our general constructions, yielding new code families and encapsulating codes previously defined in the literature. In doing so, we will implement the transitions shown in the diagram in \cref{fig: relations}.

\subsection{Fock state and spin codes from existing PI codes}
\subsubsection{Fock state codes}

Starting from \cref{prop:equivalence}, we can construct numerous new two-mode Fock state codes by leveraging existing families of PI codes. Certain families of two-dimensional PI qubit codes have received considerable attention in the literature.
In particular, Ruskai and Pollatsek \cite{ruskai-polatsek,ruskaiExchange} introduced a family of PI codes that can correct a single error. 
Ouyang \cite{ouyangPI} later constructed a family of qubit PI codes that can correct an arbitrary number of errors. The codes he introduced are defined by integer parameters $g,n$, and $u$, and hence are called $gnu$ codes. 
Reducing the number of physical qubits (the code length) improves the efficiency of the codes, and the shortest
$t$-error-correcting codes from this family have length $(2t+1)^2$. Generalizing this approach, in \cite{aydin2023family}, we introduced another family of combinatorial PI codes correcting $t$ errors, which yields shorter PI codes. In particular, a subclass of codes from that family, \cite[Thm.~5.3]{aydin2023family}, requires $(2t+1)^2-2t$ physical qubits to correct $t$ errors. 
Combining this construction with \cref{prop:equivalence} yields the following family of two-dimensional, two-mode Fock state codes.

\begin{construction}\label{constructionGmdeltaBoson}
 Let $g,m,\delta $ be nonnegative integers and let $\epsilon\in\{-1,+1\}$. Define a two-mode Fock state code $\cQ^{(b)}_{g,m,\delta,\epsilon}$ via its logical computational basis
    \begin{align*}
        &\ket{\bfc_0} =\sum_{\substack{\text{$l$ {\rm even}}\\0\leq l \leq m}} \gamma b_l\ket{gl,n-gl}_b + 
        \sum_{\substack{\text{$l$ {\rm odd}}\\0\leq l \leq m}} \gamma b_l\ket{n-gl,gl}_b,\\
        &\ket{\bfc_1} = \sum_{\substack{\text{$l$ {\rm odd}}\\0\leq l \leq m}} \gamma b_l\ket{gl,n-gl}_b +\epsilon
        \sum_{\substack{\text{$l$ {\rm even}}\\0\leq l \leq m}} \gamma b_l\ket{n-gl,gl}_b,
    \end{align*}
    where $ n=2gm+\delta+1, $
    $
    b_l=\sqrt{{\binom{m}{l}}/{\binom{n/g-l}{m+1} }},
    $
and $ \gamma =  \sqrt{\binom{n/(2g)}{m} \frac{n-2gm}{g(m+1)} } $ is the normalizing factor. 
\end{construction}

\begin{theorem}\label{theoremGMdeltaBoson}
   Let $ t $ be a nonnegative integer and let $ m\geq \left\lceil\frac{t}{2}\right\rceil $ and $\delta\geq t$. If
    $$
(   g\ge t, \epsilon=-1) \text{ or }(g\ge t+1,\epsilon=+1),
   $$
then the code $\cQ^{(b)}_{m,l,\delta,\epsilon}$ has bosonic distance $d_b=t+1$ and total excitation $N=2gm+\delta+1$.
\end{theorem}
\begin{proof}
    By \cref{prop:equivalence} and Proposition 5.5 in \cite{aydin2023family}.
\end{proof}

Let us compare this construction with existing results, in particular, those of Bergman--Van Loock \cite{VanLoock}. 
The codes they construct have total excitation $N=(t+1)^2$, which we improve to $N=(t+1)^2-t$ for any odd number of errors.
For even $t$, the parameters of the two proposals coincide.

We conclude this section with a few examples obtained using Construction \ref{constructionGmdeltaBoson}.
\begin{example}\label{example:N7Fockstate}
    Suppose $g=\delta=2$, $m=1$, and $\epsilon=-1$. Then the code $\cQ^{(b)}_{g,m,\delta,\epsilon}$ with the basis states
    \begin{align*}
        &\ket{\bfc_0} = \sqrt{\frac{3}{10}}\ket{0,7}_b + \sqrt{\frac{7}{10}}\ket{5,2}_b\\ 
        &\ket{\bfc_1} = \sqrt{\frac{7}{10}}\ket{2,5}_b-\sqrt{\frac{3}{10}}\ket{7,0}_b
    \end{align*}
    with total excitation $N=7$ has bosonic distance $d_b=3$.
\end{example}

\begin{example}
    Suppose $g=\delta=4$, $m=2$, and $\epsilon=-1$. Then the code $\cQ^{(b)}_{g,m,\delta,\epsilon}$ with the basis states
    \begin{align*}
        &\ket{\bfc_0} = \sqrt{\frac{5}{68}}\ket{0,21}_b + \sqrt{\frac{7}{12}}\ket{8,13}_b+\sqrt{\frac{35}{102}}\ket{17,4}_b\\ 
        &\ket{\bfc_1} = \sqrt{\frac{35}{102}}\ket{4,17}_b-\sqrt{\frac{7}{12}}\ket{13,8}_b-\sqrt{\frac{5}{68}}\ket{21,0}_b
    \end{align*}
    with total excitation $N=21$ has bosonic distance $d_b=5$.
\end{example}

Taking $\epsilon=+1, g=t+1$ and $m=t$, we obtain Fock state codes from $gnu$ codes relying on \cref{prop:equivalence}.  
Consider the following example, which reproduces Example 7 from \cite{chuangADCode}.
\begin{example}\label{example:9qubit}
    Suppose $g=3$, $m=1$, $\delta=2$ and $\epsilon=+1$. Then the code $\cQ^{(b)}_{g,m,\delta,\epsilon}$ with the basis states
    \begin{align*}
        \ket{\bfc_0}=\frac{1}{2}\ket{9,0}_b+\frac{\sqrt{3}}{2}\ket{3,6}_b\\
        \ket{\bfc_1}=\frac{\sqrt{3}}{2}\ket{6,3}_b+\frac{1}{2}\ket{0,9}_b
    \end{align*}
    has total excitation $N=9$ and can correct $2$-AD errors. 
    This is the two-mode Fock-state version of the \(((9,2,3))\) Ruskai code~\cite{ruskaiExchange}.
\end{example}

 For another example, an explicit construction of $K$-dimensional qudit PI codes with alphabet size $q$ is introduced in \cite[Thm.~5.2]{ouyangQudit}. In this construction, the authors present distance $d=t+1$ codes of length $N\geq(K-1)(t+1)^2$. Using this code construction, we can obtain explicit $K$-dimensional $q$-mode Fock state codes with a total excitation of at least $N=(K-1)(t+1)^2$ by \cref{prop:equivalence}. 

 \begin{example}
  We construct a $3$-dimensional Fock state code using the mapping $f$ in \cref{eq:mappingf} and the code in \cite[Example~6.4]{ouyangHigherDimensions}. This $2$-mode Fock state code, with the basis states
 \begin{align*}
     &\ket{\bfc_0} = \frac{1}{3}\ket{18,0}_b+\frac{\sqrt{7}}{3}\ket{9,9}_b+\frac{1}{3}\ket{0,18}_b\\
     &\ket{\bfc_1}=\frac{\sqrt{3}}{3}\ket{15,3}_b+\frac{\sqrt{6}}{3}\ket{6,12}_b\\
     &\ket{\bfc_2}=\frac{\sqrt{6}}{3}\ket{12,6}_b+\frac{\sqrt{3}}{3}\ket{3,15}_b~,
 \end{align*}
has distance $d_b=3$ and total excitation $N=18$.
    
 \end{example}
 
\subsubsection{Spin codes}
Considerations of the previous section apply to spin codes as well. In particular, \cref{theoremGMdeltaBoson} gives a family of $(N,2,2,t+1)$ spin codes. Importantly, we can also construct
$K$-dimensional spin codes hosted by irreps of $SU(q)$ for $K,q>2$, 
which yields a new way to encode quantum information in such few-level systems that does not rely on \(SU(2)\).
For instance, the qudit PI codes constructed in \cite{ouyangQudit}
produce a family of spin codes for arbitrary $K$ and $q$. All the examples presented in the previous section can also be cast as spin codes on discrete simplices via the mapping \(|\overline{n}\rangle_b \to |\overline{n}\rangle_s\).

\subsection{Covariant Fock state codes\\
~~~~~from PI and spin codes}\label{sec: JordanSchinger}

Identification of all three systems --- pure ``spin'' systems, Dicke-state spaces, and constant-excitation Fock-state spaces --- with a discrete simplex \(\cS_{q,N}\) allows us to inter-convert states and codes between any pair.
Furthermore, this correspondence also allows inter-conversion of certain logical operations. 

Basis states labeled by simplex points are in one-to-one correspondence with monomials in \(q\) variables of degree \(N\), i.e., the space Sym\(^N(\mathbb{C}^q)\).
As introduced in Sec.~\ref{subsec:js}, each such space admits an irrep of \(SU(q)\).
This irrep can be used to construct gates for its corresponding codes.

In the case of Fock state codes, group transformations are done by passive linear-optical transformations, whose Lie algebra is expressed by quadratic bosonic operators via the JS map~\eqref{eq:js_lie-algebra}.

If we instead switch to the labeling of Dicke states of \(N\) qudits of dimension \(q\), then we know that \(SU(q)\) transformations on this subspace are generated by the ``global'' qudit Lie algebra.
Letting \(J\) be some element of the Lie algebra in the fundamental irrep, its global representation is a sum of the \(N\) local terms, i.e., \(\sum_{i=1}^N \hat J^{(i)}\), where each qudit is acted on by its local generator 
$$ \hat J^{(i)} = \underbrace{I\otimes\dots\otimes I}_{i-1}\otimes~ J \otimes
       \underbrace{I\otimes\dots\otimes I}_{N-i}~.$$
       
Since both the Fock and Dicke labels are labeling the same irrep, we have the correspondence
\begin{equation}\label{eq:equivalence-lie}
   \left.\frac{1}{\sqrt{N}}\sum_{i=1}^{N}\hat  J^{(i)}\right|_{\text{Sym}^{N}(\mathbb{C}^{q})}\cong\left.\sum_{j,k=1}^{q}\ah_{j}^{\dagger}J_{jk}\ah_{k}\right|_{\text{Sym}^{N}(\mathbb{C}^{q})}
\end{equation}
when both representations are restricted to the same \(\text{Sym}^{N}(\mathbb{C}^{q})\) irrep of \(\mathfrak{u}(q)\).
This allows us to take any set of transversal (i.e., tensor-product) gates on a PI code and convert them to act as passive linear-optical transformations on its corresponding Fock state code, and visa versa.

Let $G$ be a subgroup of $SU(q)$, and let $\lambda$ be an irreducible representation of $G$. 
We define group elements \(g \in G\) as \(q\)-dimensional matrices represented by the fundamental irrep of \(SU(q)\).
We say that a PI code is $G$-{\em covariant} if the global (tensor-product or transversal) \(SU(q)\) representation, $g^{\otimes N}$, implements logical $\lambda(g)$ on the code for each $g\in G$. 
In other words, the action of $g^{\otimes N}$ preserves the code space of a $G$-covariant PI code for each $g\in G$, and the code space transforms as the \(\lambda\) irrep. 

We can define a $G$-covariant Fock state code in a similar manner.
For these codes, we use the passive linear optical representation of a group element \(g\),
$$D(g)=\exp\left(i\sum_{j,k=1}^{q}a_{j}^{\dagger} M_{jk}(g)a_{k}\right)~,$$
where \(M(g)\) is the Lie algebra element satisfying \(g = e^{iM}\).
A Fock state code is covariant if the physical transformation \(D(g)\) realizes the logical transformation \(\lambda(g)\) on the codespace.

We summarize the discussion so far as follows.
\begin{proposition}
    Let $G$ be a subgroup of $SU(q)$. The mapping $f$ defined in \eqref{eq:mappingf}
    sends a $G$-covariant $(N,K,q,t+1)$ PI code to a $G$-covariant $(N,K,q,t+1)$ Fock state code.
\end{proposition}
\begin{proof}
    Conservation of $G$-covariance and the distance follow by \eqref{eq:equivalence-lie} and \cref{prop:equivalence}, respectively.
\end{proof}

In \cite{IanEricDihedralIEEE}, the authors constructed PI codes that are $BD_{2b}$-covariant. Using their 
construction, we obtain the Fock state code in the following example.

\begin{example}
    The two-mode Fock state code defined by the basis states
     \begin{align*}
        &\ket{\bfc_0}=\frac{\sqrt{5}}{4}\ket{0,11}_b+\frac{\sqrt{11}}{4}\ket{8,3}_b\\
        &\ket{\bfc_1}=\frac{\sqrt{11}}{4}\ket{3,8}_b+\frac{\sqrt{5}}{4}\ket{11,0}_b
    \end{align*}
    can correct $2$ AD errors and implements all gates from the group $BD_8=\langle X,T\rangle$ using passive linear optics.
\end{example}

The \(\lambda\)-twisted \(t\)-group PI codes of Refs.~\cite{kubischtaUnitaryt,kubischta2025quantum} are \(\lambda\)-covariant.
Moreover, these codes have automatic error protection due to the nature of their particular irreps.
Any \(\lambda\)-twisted unitary \(t\)-group \(G \subset SU(q)\) admits PI codes of distance \(t+1\) inside its \(\lambda\)-irreps~\cite{kubischtaUnitaryt}. 
The distance defined in said papers is the same as the PI distance defined in our work, and the distance of the output Fock state code is the same as that per Prop.~\ref{prop:equivalence}.

The works~\cite{kubischtaUnitaryt,kubischta2025quantum} subsume the earlier work \cite{exoticGates} on covariant PI codes for \(G = 2I\), the binary icosahedral group.
The 7-qubit binary icosahedral PI code converts to the \(N=7\) two-mode Fock state code from Example~\ref{example:N7Fockstate}.
Our \(q>2\) extension yields new covariant codes defined on three or more modes.

\begin{example}
The group \(\Sigma(360\phi) \subset U(3)\) is a \(\chi_4\)-twisted 1-group, meaning that subspaces defined by its irrep \(\chi_4\) are automatically \(\chi_4\)-covariant PI codes with distance two~\cite{kubischtaUnitaryt}. 
Our results can convert these to distance-two Fock state codes that are \(\chi_4\)-covariant with respect to the passive linear-optical representation of \(\Sigma(360\phi)\).
The smallest such PI code is a $(5,3,3,2)$ code, i.e., a five-qutrit code encoding a single logical qutrit
 (a \(((5,3,2))_3\) PI code in standard notation).
This converts to a \((q=3)\)-mode 3D Fock state code in the space of Fock states of total excitation \(N=5\).
\end{example}

Along similar lines, the PI codes of Ref.~\cite{UySUd} can be mapped into Fock state codes.

Our mapping allows us to map codewords and logical operations between any pair of spaces by simply relabeling them, but we can preserve code distances only when going from PI to Fock codes or spin codes.
The reverse spin-to-Fock map is not guaranteed to preserve the code distance for general \(q\).
However, it has been shown for the \(SU(2)\) case by using Dicke states as intermediaries. 

Starting with a basis state of a spin-$J$ system, construct a 2-mode Fock state using the mapping
 \begin{align}\label{eq:mappingSpintoFock}
        \ket{J,m}\overset{\sigma^{-1}}\mapsto \ket{D_{(J+m,J-m)}}\overset{f}\mapsto \ket{J+m,J-m}_b,
    \end{align}
where the two component mappings, $\sigma^{-1}$ and $f$, are defined in 
\cref{eq:eqMapSpinPI} and \cref{eq:mappingf}, respectively.
\begin{proposition}\label{prop: Spin to Fock}
    Let $G$ be a subgroup of $SU(2)$. The composite mapping $\sigma^{-1}\circ f$ defined in \eqref{eq:mappingSpintoFock}
    sends a $G$-covariant $(N,K,2,t+1)$ spin code to a $G$-covariant $(N,K,2,t+1)$ Fock state code.
\end{proposition}
\begin{proof}
Jordan-Schwinger map sends $G$-covariant spin codes to $G$-covariant Fock state codes. Furthermore, 
\cref{lemma:ericandian} and \cref{prop:equivalence} used in succession establish the isometry claim in the statement.
\end{proof}
The authors of \cite{gross,gross2} constructed spin codes that are $G$-covariant for $G=2O$ (the binary octahedral, or Clifford, group), 2T (the binary tetrahedral group), and $2I$ (the binary icosahedral group). 
Using \cref{prop: Spin to Fock}, we can obtain two-mode Fock state codes for which logical unitaries from these subgroups of $SU(2)$ can be implemented using only beam splitters.

For example, in \cite{gross}, Gross constructed a $2O$-covariant spin $J=13/2$ code with spin distance $d_s=3$. Transforming this code into a Fock state code using the mapping in \eqref{eq:mappingSpintoFock}, we obtain a two-mode Fock state code capable of correcting $2$ AD errors. All logical gates from the Clifford group for this code can be implemented using beam splitters. The same paper introduced another $d_s=3$ spin code that is $2I$-covariant. By mapping this code to a Fock state code using \eqref{eq:mappingSpintoFock}, we obtain the code in Example \ref{example:N7Fockstate}.
Equivalently, this code can be mapped to the PI space to yield the \(2I\)-covariant PI code on 7 qubits from Ref.~\cite{exoticGates}.

\subsection{PI, Fock state, and spin codes from Tverberg partitions of \texorpdfstring{$\ell_1$}{ell1-2} codes}

In this section, we list a few examples of quantum codes constructed using partitions of classical $\ell_1$ codes. We note that
each of these examples is common to the three classes of quantum codes, which is supported by relating all the three systems
to a discrete simplex \(\cS_{q,N}\). The specific transformations between the basis elements
are given in \cref{eq:mappingf,eq:smap,eq:eqMapSpinPI}, and the coefficients $\alpha_\un$ are shared between the three expansions of the code states. We will limit our discussion to PI and Fock state codes.

\subsubsection{Fock state codes}
 Fock state codes in higher modes have previously been studied in \cite{chuangADCode,VanLoock,OuyangADCode}. In \cite{chuangADCode}, the authors constructed Fock state codes for small distances using a search algorithm. In the same paper, they introduced several examples of $K$-dimensional $q$-mode Fock state codes. The authors of \cite{OuyangADCode} studied two-dimensional Fock state codes encoded to more than two modes of Fock states. They especially constructed examples of Fock state codes with small distances that are PI. In \cite{VanLoock}, an explicit family of the $K$-dimensional Fock state code with bosonic distance $d_b>t$ is given. \remove{The code they  constructed can be defined by the basis states
 \begin{align}\label{eq:VanLoockCode}
     \ket{\bfc_r}=\frac{1}{\sqrt{K}}\sum_{j=0}^{K-1}\exp(2\pi i rj/K)\ket{t^{K,t+1}_j}^{\otimes (t+1)}.
 \end{align}
 Here
 \begin{align*}
    \ket{t^{K,a}_0}=S_d\ket{a00\ldots 0}_b,\ldots, \ket{t^{K,a}_{K-1}}=S_d\ket{(00\ldots 0a}_b,    
 \end{align*}
 where $S_d$ is a $d$-splitter. Note that the code in \eqref{eq:VanLoockCode} is a $K$-dimensional, $q=K(t+1)$-mode, distance $d_b=t+1$ Fock state code with total excitation $N=(t+1)^2$.}
 Although this construction is efficient in terms of total excitation, it does not include a large number of codes due to the restrictive structure of its construction. Note that for a $K$-dimensional $d_b=t+1$ Fock state code in \cite{VanLoock}, the number of modes is fixed to $q=K(t+1)$. However, starting with an $\ell_1$ code and using the recipe described in Section \ref{sec: l1-PI-Fock}, it is possible to construct a large number of Fock state codes for any number of modes.
Similar Fock state code construction methods are introduced in references \cite{OuyangADCode, chuangADCode}. However, our improved existence bound \eqref{eq: LowerBound} suggests searching for an $\ell_1$ code with a smaller cardinality, which potentially yields more efficient Fock state codes.

To find an explicit code, we need to find a Tverberg partition and solve the linear system in \cref{eq:conditionC4simplified1}. This task
is easy when $K=2$ and studied in \cite{OuyangADCode}. As above, start with an $\ell_1$ code $B$ of total norm $N$ with $d_1(B)\ge t+1$ and consider the system of equations
   \begin{align*}
        \sum_{\uh\in\cC} \frac{\binom{N-t}{\uh-\ue}}{\binom{N}{\uh}}y_\uh=0\quad \text{for all $\ue\in\cS_{q,t}$}.
    \end{align*}
Solving it for $(y_\uh)$, we define 
         \begin{equation}\label{eq:conditionC4K2}
        x_\uh=
     \begin{cases}
            y_\uh,\quad \text{if $y_\uh>0$}\\
            -y_\uh,\quad \text{if $y_\uh<0$}.
        \end{cases}
    \end{equation}
Using \cref{prop:bound}, we can construct Fock state codes as shown in the following examples.
\begin{example}\label{example:N3Code}
    Let us take $N=q=3$ and $t=1$. Consider an $\ell_1$ code given by
    \begin{align*}
        B=\{ (3,0,0),(0,3,0),(0,0,3),(1,1,1)\}
    \end{align*}
    Observe that $|B|=4$ and $d_1(B)=2$, which matches the lower bound \eqref{eq: LowerBound} and implies that
    there exists a bosonic code with total excitation $N=3$ and distance $d_b=2$. Let us construct it explicitly.
    The matrix of coefficients of the system for $(y_\uh)$ 
    \begin{align*}
        \begin{bmatrix}
            1&0&0&\frac{1}{3}\\
            0&1&0&\frac{1}{3}\\
            0&0&1&\frac{1}{3}
        \end{bmatrix}
    \end{align*}
 yields a solution $(1/3,1/3,1/3,-1)$. This in turn produces a Fock state code with the basis
    \begin{align*}
        &\ket{\bfc_0} = \sqrt{\frac{1}{3}}(\ket{3,0,0}_b + \ket{0,3,0}_b + \ket{0,0,3}_b) \\  
        &\ket{\bfc_1}=\ket{1,1,1}_b
    \end{align*}
that has bosonic distance $d_b=2$ and corrects a single AD error. Observe that this example
recovers the Wasilewski-Banaczek code \cite{Wasilewski}.
\end{example}

\begin{example}
    \begin{figure}[ht]
    \includegraphics[width=.9\linewidth]{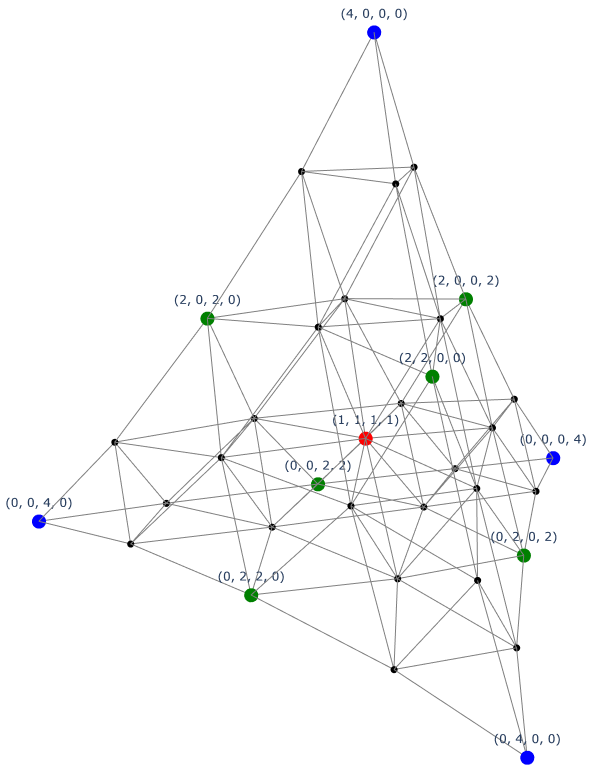}
    \caption{The simplex $\cS_{4,4}$. The colored set of 11 vertices represents a $3$-dimensional code $B$ with $\ell_1$ distance 2 in the simplex. Each color shows a subset of code points $\bfx$ of fixed composition $\comp(\bfx)$. 
There are 3 such subsets, $B_0,B_1,B_2$, which form a Tverberg partition ${\mathscr I}_{11,3}$
of the codeword set with respect to the set of points/equations of the form \eqref{eq: K system} with coefficients given in \eqref{eq: coefficients}.} \label{fig: S44}
    \end{figure}
Let us consider a more complicated case of larger $K$. In this example we construct an $(N=4,K=3,q=4,d_b=2)$ Fock state code. According to our recipe above, we begin with an
$\ell_1$ code $B$ of size 11 and distance 2, shown in \cref{fig: S44} as a colored subset of $\cS_{q,N}$. Note that $|B|=11$, hence obeying the lower bound in \eqref{eq: LowerBound}. Consider the following partition of the code $B$:
\begin{align*}
    B_0&=\{4000,0400,0040,0004\}\\
    B_1&=\{2200,2020,2002,0220,0202,0022\}\\
    B_3&=\{1111\}.
\end{align*}
The code corrects any single $\ell_1$ error, i.e., any vector from the set $1000,0100,0010,0001$. 
To find the coefficients of the code basis, we solve the equations
  $$
  \sum_{\uc\in C_0}\frac{\binom{N-t}{\uc-\ue}}{\binom N\uc}x_{\uc}=\sum_{\uc \in C_1}
  \frac{\binom{N-t}{\uc-\ue}}{\binom N\uc}x_{\uc}=
  \sum_{\uc \in C_2} \frac{\binom{N-t}{\uc-\ue}}{\binom N\uc}x_{\uc}
  $$
for all $\ue\in \cS_{4,1}$. This yields
   $$
   x_{\uc}=\begin{cases}
       1/4&\text{if }\uc\in B_0\\
       1/6&\text{if }\uc\in B_1\\
       1&\text{if }\uc\in B_2.
   \end{cases}
  $$  
\end{example}
The basis vectors of the target Fock state code take the form
\begin{align*}
    \ket{\bfc_0}&=\frac12(\ket{4,0,0,0}_b+\ket{0,4,0,0}_b+\ket{0,0,4,0}_b+\ket{0,0,0,4}_b\\
    \ket{\bfc_1}&=\frac1{\sqrt 6}(\ket{2,2,0,0}_b+\ket{2,0,2,0}_b+\ket{2,0,0,2}_b+\ket{0,2,2,0}_b\\
    &\hspace*{.2in}+\ket{0,2,0,2}_b+\ket{0,0,2,2}_b)\\
    \ket{\bfc_2}&=\ket{1,1,1,1}_b.
\end{align*}
Generally speaking, finding a Tverberg partition for a given set of points has polynomial complexity if the dimension $m$ is fixed and is difficult otherwise \cite{barany2018tverberg}. However, in small examples such as this one, it can be found with few complications.

\subsubsection{Qudit PI codes}

 We can also construct new qudit PI codes using Tverberg partitions of $\ell_1$ codes as in the previous section that dealt with Fock state codes. The following example is constructed using the same Tverberg 
 partition as in \cref{example:N3Code}.
\begin{example}\label{example:3qutrits}
    The PI code defined by the basis
    \begin{align*}
        &\ket{\bfc_0} = \sqrt{\frac{1}{3}}\left(\ket{000}+\ket{111}+\ket{222}\right)\\
        &\ket{\bfc_1}=\sqrt{\frac{1}{6}}\left(\ket{012}+\ket{021}+\ket{102}+\ket{120}+\ket{201}+\ket{210}\right)
    \end{align*}
    has alphabet size $q=3$, length $N=3$, and distance $d=2$. Note that it is shorter than all previously known PI codes with distance $d=2$.
Observe that this code is the \([[3,1,2]]_3\) three-qutrit stabilizer code \cite{cleve1999share} projected into the permutation-invariant subspace of three qutrits, i.e., 
the trivial irrep of \(S_3\) acting by permutations on three factors of \(\mathbb{C}^3\).
\end{example}

\begin{example}\label{example:PI-N6}
    Let us take $N=q=6$ and $t=2$. Let $B_1=(1^6), B_2=\pi{(3^20^4)}, B_3=\pi(60^5)$,
    where $\pi(\cdot)$ denotes the set of tuples formed of all permutations of the argument. Consider an $\ell_1$ code 
$
        B=\bigcup_{i=1}^3 B_i.
$
The code $B$ has size $22$, which meets the bound \eqref{eq: LowerBound} for $K=2$ with equality. Therefore, there must be a qudit PI code with distance $d=3$. Following the same steps as in \cref{example:N3Code}, we obtain the PI code defined by the basis
    \begin{align*}
        &\ket{\bfc_0}=\sqrt{\frac{1}{15}}\sum_{\un \in B_3}\ket{D_\un} + \sqrt{\frac{3}{5}}\ket{D_{(111111)}},\\
        &\ket{\bfc_1}=\sqrt{\frac{1}{15}}\sum_{\un \in B_2}\ket{D_\un}.
    \end{align*}
\end{example}
Previously, the shortest PI code that can correct a single error has alphabet size $q=2$ and length $N=7$. Observe that the code over the alphabet of size $q=6$ in the example above has length $N=6$, which is less than the best known PI code with similar error correction properties. Furthermore, mapping the PI code in Example \ref{example:PI-N6} to a Fock state code using \eqref{eq:mappingf}, we recover the Fock state code in Example 2 in \cite{OuyangADCode}.

\subsubsection{Existence of efficient \texorpdfstring{$SU(q)$}{} codes}
We conclude this section by formulating general existence conditions for the three types of codes.
In addition to yielding more flexibility, codes obtained by way of $\ell_1$ codes may also produce more efficient codes in terms of the length (PI codes), total excitation (Fock state code), or total spin (spin codes). 
\begin{theorem}\label{theorem:existence}  Let $\cQ(N,K,N,d)$ be one of {\rm \{PI code, spin code, Fock state code\}}.
     Let $K,t\geq 2$ be integers. Then there exists an explicitly constructible $K$-dimensional Fock state code $\cQ$ with $\{$length, total excitation, total spin$\}$ $N=(K-1)t(t+1)$, and distance $d=t+1$. 
\end{theorem}
\begin{proof}
    By combining \cref{prop:bound} and \cref{theorem:existence l1}. \qedhere
\end{proof}
Note that when $K=2$, \cref{theorem:existence} states that there exists Fock state codes with distance $d_b=t+1$ and total excitation $N= t(t+1)$. The best previously known codes  \cite{VanLoock} have total excitation $N=(t+1)^2$. Hence, our lower bound for the total excitation is $t+1$ less than the previously best known bound for a code with distance $d_b=t+1$.

In \cite{OuyangADCode}, the authors constructed explicit $2$-dimensional codes which saturates the bound in \cref{theorem:existence} for $t=2,3,4,5$. Here we show that there exists a two-dimensional code with $N=t(t+1)$ for any $t\geq 2$. Furthermore, we extend the result from $2$-dimensional codes to $K$-dimensional codes in general.

Previously best known lower bound for the length of a qudit $K$-dimensional PI code with distance $d=t+1$ is $N\geq (K-1)(t+1)^2$ \cite{OuyangADCode}. Therefore, the result of \cref{theorem:existence} improves
the best previously known result for the parameters of qudit PI codes.

\section{Concluding remarks}

We unify and extend a broad class of codes compatible with nuclear manifolds of atomic and molecular systems, the permutation-invariant (PI) space of multiple qudits, and the constant-excitation Fock subspace of multiple bosonic modes.
This unification is possible because basis states for all three state spaces are in one-to-one correspondence with points in the discrete simplex, and noise models on all three state spaces are related by the fact that all three spaces house irreps of the Lie group \(SU(q)\) for some \(q\).
This unification extends previous results for \(q=2\), and we construct new codes that are simultaneously applicable to all three state spaces using classical codes on the discrete simplex and partitioning results from convex geometry. 
In particular, this yields examples
of code sequences with improved asymptotics of the number of correctable errors in all the three spaces.

We conclude by discussing further extensions.

{\bf Approximate codes: }
Recall our use of Tverberg partitions for proving the error correction conditions for Fock state codes and PI codes in \cref{prop:bound} (also \cref{Theo:BosonMainTheoremAD}), where the common point in the convex hulls of the point
subsets yielded a set of basis coefficients that satisfy the error correction requirements in \hyperlink{C3}{(C3)},
\hyperlink{(C4)}{(C4)}. For such a point to exist, we need to limit the dimension $m$ of the real space
that hosts the Tverberg points, and this constrains the scaling of the code parameters (rephrasing, given $m$ and $K$, we need to have sufficiently many points). Since the inequality in the
Tverberg theorem is known to be tight, removing the assumption on $m$ is not possible without changing the statement.
It turns out that if, instead of seeking that the convex hulls have a common point, we require only that they all intersect
a fixed ball of small radius, then the dimension can be dropped from the statement of the theorem \cite{adiprasito2020theorems}. For the coding problem at hand,
this would result in some version of approximate error correction, addressing the number of errors larger than we are able to
guarantee with exact recovery. There are further conditions to be met to implement this plan, and we leave this as an interesting future direction. The general idea of approximate error correction has been discussed in earlier works \cite{beny2010general,brandao2019quantum,gschwendtner2019quantum} as well as in \cite{movassagh2024constructing},
however, the approach presented here suggests a concrete direction for potentially constructing such codes.

{\bf Code symmetrization: }
Conventional multi-qudit block codes are generally not PI, but it was noticed early on that a well-known 9-qubit PI code~\cite{ruskaiExchange} can be obtained from the Shor 9-qubit code by projecting the latter into the PI subspace of 9 qubits (see \cref{example:9qubit}).
Similarly, we identify a PI/Fock/spin code that is a projection of the \([[3,1,2]]_3\) three-qutrit stabilizer code onto the PI subspace (see. \cref{example:3qutrits}).
Projecting other established non-PI stabilizer codes will likely yield interesting and unique PI codes, which in turn may be convertible into spin or Fock state codes via our mapping.
It would be interesting to determine when such projections and mappings preserve the distances of the original stabilizer codes.

\begin{acknowledgements}

V.V.A.~acknowledges Andrea Morello for stimulating discussions. The research of A.A. and A.B. was partially supported by NSF grant CCF-2330909. V.V.A.~acknowledges NSF grant OMA2120757 (QLCI).
\end{acknowledgements}

\bibliography{PI}
\appendix

\onecolumngrid

\section{Multinomial coefficients}
A multi-dimensional version of the Vandermonde convolution, \cref{lemma:convolution}, is slightly less known than its one-dimensional analog. Since we need it in the main text, in this appendix, we present it together with a short proof. We start with the following
\begin{lemma}\label{Lemma: multinomial}
For integer $N\geq t\geq 0$, let $\un,\ue$ be $q$-tuples with $\sum_{i=0}^{q-1} n_i=N,\sum_{i=0}^{q-1} e_i=t$
Then
    \begin{align}\label{eq: md}
       \displaystyle \frac{\binom{n_0}{e_0}\binom{n_1}{e_1}\ldots\binom{n_{q-1}}{e_{q-1}}}{\binom{N}{t}}=\frac{\binom{t}{\ue}\binom{N-t}{\un-\ue}}{\binom{N}{\un}}.
    \end{align}
    \begin{proof}
        By a direct calculation,
        \begin{align*}
            \frac{\binom{N-t}{\un-\ue}}{\binom{N}{\un}}
            &=\frac{(N-t)!}{\prod_{i=0}^{q-1}(n_i-e_i)!}\frac{\prod_{i=0}^{q-1}n_i!}{N!}=\frac{e_0!\ldots e_{q-1}!}{t!}\frac{t!(N-t)!}{N!}\prod_{i=0}^{q-1}\frac{n_i!}{(n_i-e_i)!e_i!}\\
             &=\frac{1}{\binom{t}{\ue}\binom{N}{t}}\prod_{i=0}^{q-1}\binom{n_i}{e_i}. \qedhere
            \end{align*}
    \end{proof}
\end{lemma}
\begin{lemma}\label{lemma:convolution}
    Let $N$, $t$ be nonnegative integers with $N\geq t$, and $\un \in \cS_{q,N}$. Then
    \begin{align*}
        \sum_{\ue \in \cS_{q,t}}\binom{t}{\ue}\binom{N-t}{\un-\ue} = \binom{N}{\un}.
    \end{align*}
\end{lemma}
\begin{proof}
    The standard (one-dimensional) Vandermonde convolution has the form
     $$
     \sum_{e=0}^t \binom {n_0}e\binom{n_1}{t-e}=\binom{n_0+n_1}{t}.
     $$
By induction, we quickly conclude that
   $$
   \sum_{\ue\in \cS_{q,t}}\prod_{i=0}^{q-1}\binom {n_i}{e_i}=\binom Nt,
   $$
so the left-hand side of \cref{eq: md} is 1 when summed on $\ue\in\cS_{q,t}$. Then so is the right-hand side, which is the claim of the lemma.
\end{proof}

\twocolumngrid

\end{document}